\documentclass[useAMS,referee,usenatbib]{biom}

\usepackage[figuresright]{rotating}
\usepackage{url}
\usepackage{amsmath}
\usepackage{dsfont}
\usepackage{amssymb}
\usepackage{hhline}
\usepackage{multirow}
\usepackage{xcolor}
\usepackage{subfig}
\usepackage{adjustbox}
\usepackage{caption}
\usepackage{changepage}
\usepackage{mathrsfs}
\usepackage{natbib}
\usepackage{enumitem}
\usepackage{hyperref}

\newcounter{as}
\newcommand{\assum}[1]{%
  \refstepcounter{as}\label{#1}%
  \par\medskip\noindent\hangindent=2em\hangafter=1
  (\theas)\hspace{0.5em}\ignorespaces}

\DeclareMathOperator*{\ind}{1{\hskip -2.5 pt}\hbox{I}}

\newcommand{\R}{\mathds{R}}

\newcommand{\bt}{\bm{t}}

\newcommand{\bw}{\bm{w}}

\newcommand{\bz}{\bm{z}}

\newcommand{\bD}{\bm{D}}

\newcommand{\bU}{\bm{U}}
\newcommand{\bV}{\bm{V}}
\newcommand{\bW}{\bm{W}}

\newcommand{\bZ}{\bm{Z}}

\newcommand{\balpha}{\bm{\alpha}}

\newcommand{\bkappa}{\bm{\kappa}}

\newcommand{\bPsi}{\bm{\Psi}}

\newcommand{\bbeta}{\mbox{\boldmath $\beta$}}
\renewcommand{\hat}{\widehat}

\newcommand{\btheta}{\bm{\theta}}
\newcommand{\E}{\mathds{E}}
\renewcommand{\P}{\mathds{P}}
\def\T{{\sf T}}

\title[U-Shaped Risk]{Rank-Based Estimation of U-Shaped Biomarker Risk Curves and Critical Points for Time-to-Event Outcomes}

\author{Zhirui Fu$^{1}$, 
Mei-Cheng Wang$^{2}$,
Yu Du$^{3}$, and
Yuxin Zhu$^{4, 2,*}$\email{daisy@jhu.edu}\\
$^{1}$Department of Biostatistics and Bioinformatics, Emory University Rollins School of Public Health\\
$^{2}$Department of Biostatistics, Johns Hopkins Bloomberg School of Public Health\\
$^{3}$Eli Lilly and Company\\
$^{4}$Department of Neurology, Johns Hopkins School of Medicine}

\makeatletter
\def\@journal{}
\def\ps@titlepage{\let\@mkboth\@gobbletwo
 \def\@oddhead{}%
 \def\@oddfoot{}%
 \def\@evenhead{}%
 \def\@evenfoot{}%
 \def\sectionmark##1{}%
 \def\subsectionmark##1{}}
\def\ps@headings{\let\@mkboth\markboth
 \def\@oddhead{\Large\hfill{\small\it\@shorttitle}\hfill\hspace{1.5em}%
   \rm\@ddell\hbox to 0pt{\small\thepage}}%
 \def\@oddfoot{}%
 \def\@evenhead{\Large\@ddell\hbox to 0pt{\small\thepage}\hspace{1.5em}\hfill}%
 \def\@evenfoot{}%
 \def\sectionmark##1{\markboth{##1}{}}%
 \def\subsectionmark##1{\markright{##1}}}
\ps@headings
\makeatother

\begin{document}

\label{firstpage}

\begin{abstract}
U-shaped relationships between prognostic biomarker levels and adverse event risk are commonly observed across diseases, where both low and high biomarker values are associated with elevated risk, with a well-defined minimum---the critical point---marking the biomarker value of the lowest risk. The U-shaped risk curve, especially the location of the critical point, informs the identification of high- and low-risk subgroups. However, existing methods are limited: U-shaped risk models rarely accommodate survival outcomes, and existing survival analysis methods do not enable estimation of or formal inference for the critical point. To fill this gap, we propose a semiparametric transformation model that explicitly parameterizes the critical point, a rank-based maximum C-index estimator for the parametric component, and a smoothed Kaplan-Meier estimation approach for the nonparametric component. The resulting framework estimates both subgroup-specific U-shaped risk curves and their critical points within a single survival model. We establish consistency and asymptotic normality of the proposed estimators and demonstrate their finite-sample performance through numerical studies. We apply the proposed methods to UK Biobank data to characterize subgroup-specific U-shaped associations between body mass index and all-cause mortality and to identify the corresponding critical points.
\end{abstract}


\maketitle

\section{Introduction}
\label{sec:introduction}

U-shaped risk structures are prevalent across biomedical domains. As prognostic biomarker values increase, the risk of an adverse event decreases to the lowest point (the critical point) and then increases, constituting a U-shaped curve. For example, high or low levels of high-density lipoprotein cholesterol are associated with an increased risk of stroke \citep{li2022u}; high or low levels of low-density lipoprotein cholesterol are associated with an increased all-cause mortality risk during the acute stage of ischemic stroke \citep{chen2023u}; high or low body mass index (BMI) levels are associated with an increased all-cause mortality risk \citep{bhaskaran2018association}.  These U-shaped risk curves, together with critical point and critical region estimates, identify optimal biomarker levels and enable classification of patients into risk subgroups, guiding preventive care and treatment design.
The estimation of critical points has been discussed in the broader change-point literature \citep{khodadadi2008change}, where the primary goal is to detect locations at which the underlying regression structure changes. The problem considered in this paper, however, is distinct: in U-shaped risk settings, the critical point corresponds to a biomarker-defined risk extremum rather than a structural break in the regression function. This distinction shifts the methodological focus from detecting a change in regression structure to estimating and quantifying uncertainty for the biomarker value at minimum risk, a task that remains unresolved in both non-survival and survival settings.

In non-survival settings, existing methods accommodate U-shaped risk patterns to varying degrees, but approaches to critical-point estimation and inference remain restrictive. Parametric approaches, such as piecewise linear and polynomial models, impose strong functional assumptions and either fix the critical point a priori or leave it unparameterized, precluding direct inference \citep{muggeo2003estimating}. Kernel-based nonparametric estimators \citep{muller1992change} avoid parametric forms but are sensitive to bandwidth selection, computationally intensive, and lack valid inference. Profile likelihood approaches \citep{hall2000change} improve computational efficiency but lack established asymptotic theory. Bayesian models \citep{gossl2001bayesian} instead enable model-based inference but require costly MCMC computation and prior specification. These methods do not accommodates time-to-event outcomes with censoring. 

Meanwhile in survival settings, methodological challenges are further compounded by censoring. Existing survival methods primarily estimate time- or hazard-based critical points \citep{muller1990nonparametric, brazzale2019nonparametric, zhao2009change} rather than biomarker-based ones. Spline-based Cox models require the critical point to be specified a priori as part of the spline knots, so it cannot be estimated from the data, and uncertainty quantification for tuned-knot locations remains challenging. The most closely related work \citep{tang2017biomarker} estimates a biomarker-based critical point, but it addresses a longitudinal-trajectory setting rather than censored survival outcomes, and the asymptotic properties of its estimator remain unstudied. Taken together, these limitations across both non-survival and survival settings reveal a fundamental methodological gap: direct estimation and valid inference remain unavailable for the biomarker-based critical point.

Specifically, no existing method for survival outcomes provides all three of the following capabilities simultaneously: (i) joint estimation of the critical point together with the model parameters, rather than post-hoc numerical extraction from a fitted curve; (ii) formal inference, such as confidence intervals, for the critical point location; and (iii) covariate-dependent critical points that allow different patient subgroups to have distinct biomarker thresholds for minimum risk, estimated within a single model.

These limitations point to the need for a method that characterizes U-shaped biomarker risk curves under censoring, estimates critical points with valid inference, and remains computationally efficient, robust, and flexible to accommodate complex risk patterns, including multiple turning points along the risk curve. In this work, we propose a semiparametric model for the U-shaped curve of event risk as a function of a biomarker while accounting for possible heterogeneity due to additional risk factors. Our model adopts an unspecified transformation and parameterizes the critical point on the transformed risk scale through a non-traditional piecewise linear function that naturally allows the critical point to be estimated.

The remainder of this paper is organized as follows. Section \ref{sec:model} describes our model and its properties. In Section \ref{sec:estimation-finite}, we present a rank-based estimation approach and establish the asymptotic properties of the proposed estimator. We introduce a smoothed Kaplan-Meier estimator for our model's nonparametric transformation component in Section \ref{sec:estimation-infinite}. Section \ref{sec:estimation-critical-point} presents the estimation and inference for the critical point (i.e., the lowest point of the U-shaped curve) and the critical region (i.e., the range of biomarker values corresponding to low risks). Simulation studies in Section \ref{sec:simulation} evaluate our methods' finite-sample performance. In Section \ref{sec:data-analyses}, we apply our methods to data from the UK Biobank, a large prospective cohort study with long-term follow-up, to demonstrate their performance and practical relevance in real-world survival analysis.

\section{Semiparametric Model with a Flexible Critical Point}\label{sec:model}

In this section, we propose a semiparametric model for the time-to-event outcome $T$ as a function of a biomarker $X$, designed to capture U-shaped risk patterns. In addition to the effect of the biomarker on event risk, we allow covariates $\bZ \in\R^p$ to modify the U-shaped relationship, so that the critical point varies across covariate-defined subgroups. 
We assume prior knowledge that a U-shaped relationship exists between $X$ and event risk---that is, both low and high biomarker levels correspond to an elevated risk of event occurrence---at an arbitrary time point $t$. 
We leverage this structure by considering the following semiparametric model that explicitly parameterizes the critical point:
\begin{align}
T = G\big\{-\max(-X+\bZ^{\T}\balpha_1, \beta_0 + \beta_1 X + \bZ^{\T}\balpha_2) + \epsilon\big\},\label{semi_model}
\end{align}
where $\beta_0 \in \R, \beta_1 >0$ are parameters depicting the location and shape of the U-shaped risk curve, $\balpha_1, \balpha_2\in\R^p$ are covariate coefficients that allow the location and shape of the U-shaped risk curve to vary for different covariate values, $\epsilon$ is a continuous error term independent of $(X,\bZ)$ following some otherwise unspecified distribution, and $G: \R \rightarrow{\R}$ is a continuous and strictly monotonically increasing function but otherwise unspecified. Note that we normalize the coefficient of $X$ in the first component of the piecewise linear structure to 1 without loss of generality: a more general specification replacing $-X$ by $-(\beta_0^* + \beta_1^* X)$ in the first component yields an equivalent model class through rescaling, since both $G$ and the distribution of $\epsilon$ are unspecified \citep[cf.][]{han1987non}.

Model (\ref{semi_model}) captures the U-shaped relationship while explicitly parameterizing the critical point. The two linear components inside the maximum intersect at $X_c=(\beta_1 +1)^{-1}\big\{ \bZ^{\T}(\balpha_1-\balpha_2) - \beta_0 \big\}$, which defines the critical point. At this value, the expected event-free time (expectation of $T$) reaches its maximum, or equivalently, the risk of adverse events reaches its minimum.
Model (\ref{semi_model}) is equivalently written as:
\begin{align*}
    T = \ind(X<X_c)\cdot G(X - \bZ^{\T}\balpha_1 + \epsilon) + \ind(X\geq X_c) \cdot G(-\beta_0 - \beta_1X-\bZ^{\T}\balpha_2 + \epsilon).
\end{align*}
Under the constraint of $\beta_1>0$, the piecewise linear structure inside $G(\cdot)$ implies that $T$ first stochastically increases as $X$ increases on the left of the critical point, and then stochastically decreases as $X$ increases on the right of the critical point. Equivalently, as biomarker $X$ increases, the risk of adverse events first decreases and then increases, forming a U-shaped curve. 
This U-shaped risk pattern is flexible in shape due to the inclusion of an unspecified transformation function $G$ and an unspecified error term $\epsilon$, which together allow for smooth distortion of the curve. Specifically, $G$ accommodates non-linear scaling of the risk pattern across the range of $X$, enabling the U-shaped trend to take more general smooth forms, while preserving the inherent monotonicity on each side of the critical point. Notably, $\epsilon$ and $G$ are not individually identifiable. Instead, these two nonparametric components are estimated concurrently through the estimation of the survival probability as a function of the biomarker and covariates. 
Despite this nonparametric flexibility, the underlying piecewise linear structure ensures that the critical point $X_c$ remains explicitly expressed as a function of finite-dimensional parameters $(\beta_0,\beta_1,\balpha_1,\balpha_2)$. This model structure enables statistical inference for these parameters, and the explicit parametric form of $X_c$ extends the inference to the critical point, as well as related quantities such as the critical region.

Model (\ref{semi_model}) extends the class of semiparametric transformation models introduced by \cite{han1987non}, in which $T = G(X^{\T}\bbeta + \epsilon)$ with $G$ and the distribution of $\epsilon$ both unspecified. Our model replaces the linear predictor $X^{\T}\bbeta$ with the piecewise linear term $-\max(-X+\bZ^{\T}\balpha_1, \beta_0+\beta_1X+\bZ^{\T}\balpha_2)$, thereby accommodating non-monotone biomarker--outcome relationships while retaining the rank-based estimation framework that \cite{han1987non} established for monotone index models.

Several widely adopted survival and related regression models are special cases of our proposed semiparametric model (\ref{semi_model}) under appropriate choices of the transformation function $G(\cdot)$ and the error distribution $\epsilon$.

\begin{example}
    (Proportional Hazards Model)
    Taking $G(t) = \Lambda_0^{-1}(\exp(t))$ and $\epsilon$ with distribution function $F(u) =1-\exp\{-\exp\{u\}\}$, model (\ref{semi_model}) implies
    $$\lambda(t)=\lambda_0(t)\exp\{\max(-X+\bZ^{\T}\balpha_1, \beta_0 + \beta_1 X + \bZ^{\T}\balpha_2)\},$$ which corresponds to a Cox proportional hazards model with a piecewise linear index structure. Equivalently, $-\log\{-\log S(t|X,\bZ)\} = -G^{-1}(t) - \max(-X+\bZ^{\T}\balpha_1, \beta_0 + \beta_1 X + \bZ^{\T}\balpha_2)$, exhibiting an additive structure on the complementary log-log scale.
\end{example}

\begin{example}
    (Accelerated Failure Time Model)
    Taking $G(t) = \exp(t)$ and leaving the error distribution unspecified, model (\ref{semi_model}) implies
    $$\log T = -\max(-X+\bZ^{\T}\balpha_1, \beta_0 + \beta_1 X + \bZ^{\T}\balpha_2) + \epsilon,$$
    which corresponds to an accelerated failure time model with a piecewise linear index structure on the log-time scale.
\end{example}

\begin{example}
    (Proportional Odds Model)
    Taking $\epsilon$ with logistic distribution $F(u)=\{1+\exp\{-u\}\}^{-1}$ and leaving $G$ unspecified, model (\ref{semi_model}) implies
    $$\log\frac{1-S(t|X,\bZ)}{S(t|X,\bZ)}=G^{-1}(t)+\max(-X+\bZ^{\T}\balpha_1, \beta_0 + \beta_1 X + \bZ^{\T}\balpha_2),$$
    which corresponds to a proportional odds model, where the log-odds of failure by time $t$ exhibits an additive structure between a baseline time component $G^{-1}(t)$ and the covariate dependent risk index $\max(-X+\bZ^{\T}\balpha_1, \beta_0 + \beta_1 X + \bZ^{\T}\balpha_2)$.
\end{example}

\begin{example}
    (Binary Choice Model)
    Taking $G(t)=\ind\{t\geq 0\}$ as a limiting case, model (\ref{semi_model}) reduces to the latent index representation
    $$T = \ind\{-\max(-X+\bZ^{\T}\balpha_1, \beta_0+\beta_1X+\bZ^{\T}\balpha_2)+\epsilon\geq 0\}.$$
    If $\epsilon \sim N(0,1)$, the model corresponds to a probit specification with
    \[
    P(T=1|X,\bZ)=1-\Phi\big\{\max(-X+\bZ^{\T}\balpha_1, \beta_0 + \beta_1 X + \bZ^{\T}\balpha_2)\big\},
    \]
    whereas if $\epsilon$ follows a logistic distribution $F(u)=\{1+\exp\{-u\}\}^{-1}$, the model reduces to a logit model with
    \[
    P(T=1|X,\bZ)=\big\{1+\exp\big\{\max(-X+\bZ^{\T}\balpha_1, \beta_0 + \beta_1 X + \bZ^{\T}\balpha_2)\big\}\big\}^{-1}.
    \]
\end{example}

Existing approaches such as proportional hazards models with piecewise polynomial splines accommodate U-shaped relationships, but they do not parametrize the critical point. Its location must either be pre-specified as a knot or extracted post-hoc by numerically optimizing the fitted curve, and in neither case does a valid inferential framework exist for the critical point itself. Consequently, these approaches cannot provide confidence intervals or hypothesis tests for the biomarker value at minimum risk. Our model addresses this limitation by expressing the critical point as $X_c=(\beta_1 +1)^{-1}\{ \bZ^{\T}(\balpha_1-\balpha_2) - \beta_0 \}$, a closed-form function of the finite-dimensional model parameters, enabling point estimation and uncertainty quantification while imposing minimal assumptions on the transformation $G$ and error term $\epsilon$.

\section{A Maximum C-index Estimator: Rank-Based Maximization for the Parametric Component}\label{sec:estimation-finite}

In this section, we propose a rank-based maximization estimator for the finite-dimensional parametric component in model (\ref{semi_model}) by comparing the relative ordering of survival outcomes across subjects. Let $(Y_i, \Delta_i, X_i, \bZ_i)$, $i=1,\ldots,n$, denote an i.i.d. sample where $Y_i=\min(T_i,C_i)$ is the observed time, $\Delta_i=\ind(T_i<C_i)$ is the event indicator, $X_i$ is the biomarker of interest, and $\bZ_i$ represents the collection of risk factors that potentially modify the specific patterns of U-shaped risk curves. We assume conditionally independent right censoring, i.e., $C_i \perp T_i | (X_i,\bZ_i)$, so that the relative ordering information used by our objective remains valid after censoring. This assumption is sufficient but not necessary for the consistency of the proposed estimator; weaker conditions for rank-based estimators are discussed in \cite{wang2019maximum}.

We estimate the finite-dimensional parameter vector without specifying $G$ or the distribution of $\epsilon$ by exploiting the stochastic ordering implied by model ($\ref{semi_model}$). 
For simplicity, let $\btheta=(\beta_0,\beta_1,\balpha_1^{\T},\balpha_2^{\T})^{\T}\in\Theta$, and define the piecewise linear risk index $H(X,\bZ;\btheta)=\max(-X+\bZ^{\T}\balpha_1,\beta_0+\beta_1X+\bZ^{\T}\balpha_2)$. Consider any two independent observations $(Y_i, \Delta_i, X_i, \bZ_i)$ and $(Y_j, \Delta_j, X_j, \bZ_j)$.
Because $G$ is strictly increasing and $\epsilon$ is independent of $(X,\bZ)$, a larger value of $H(X,\bZ;\btheta)$ corresponds to stochastically smaller survival time: whenever $H(X_i,\bZ_i;\btheta)>H(X_j,\bZ_j;\btheta)$, we have $-H(X_i,\bZ_i;\btheta)+\epsilon_i$ stochastically smaller than $-H(X_j,\bZ_j;\btheta)+\epsilon_j$, which implies that $T_i$ is stochastically smaller than $T_j$, and hence
$P(T_i<T_j)>P(T_j<T_i)$
\citep{han1987non}.
This ordering carries through the presence of right censoring: whenever $H(X_i,\bZ_i;\btheta)>H(X_j,\bZ_j;\btheta)$, we have 
$P(Y_i<Y_j,\Delta_i=1)>P(Y_j<Y_i,\Delta_j=1)$ (proof given in Web Appendix A). Because this ordering is determined by the ranks of the risk index $H(X,\bZ;\btheta)$, $\btheta$ is identified without specifying either $G$ or the distribution of $\epsilon$.

The stochastic ordering motivates estimating $\btheta$ by maximizing the agreement between the risk index and the observed survival outcomes. A pair $(i, j)$ is comparable if the subject with the smaller observed time has $\Delta=1$. Among comparable pairs, we call a pair concordant if the subject with the smaller observed time has the larger value of $H(X,\bZ; \btheta)$.
Counting concordant pairs over all ordered pairs yields the empirical C-index $C_n(\btheta)$:
\begin{align}
\begin{split}
    C_n(\btheta) = & n^{-2}\sum_{i=1}^n\sum_{j=1}^n
    \Big[\ind(Y_i<Y_j, \Delta_i=1)\cdot 
    \ind\big\{H(X_i,\bZ_i;\btheta)
    >H(X_j, \bZ_j; \btheta)\big\}\\
    &+\ind(Y_i>Y_j, \Delta_j=1)\cdot \ind\big\{H(X_i, \bZ_i;\btheta)
 <H(X_j, \bZ_j;\btheta)\big\}\Big].
\end{split}\label{eq-cindex}
\end{align}
Expression (\ref{eq-cindex}) is a Harrell-type C-index for right-censored data \citep{harrell1982evaluating} with the linear predictor replaced by the risk index $H(X, \bZ;\btheta)$. $C_n(\btheta)$ measures how well $H(X, \bZ;\btheta)$ ranks subjects according to their observed survival times, with larger values corresponding to higher risk and earlier events. C-index values near 1 indicate near-perfect ranking, and values near 0.5 indicate random ranking. 
Maximizing $C_n(\btheta)$ yields $\widehat{\btheta}=\text{argmax}_{\btheta \in \Theta}C_n(\btheta)$, the maximum C-index estimator (MCE). 

The MCE recovers the true parameter $\btheta_0$ because the population objective $\E\{C_n(\btheta)\}$ attains its maximum at $\btheta_0$. Taking the expectation of (\ref{eq-cindex}) with respect to $(X_i,\bZ_i,X_j,\bZ_j)$ gives:
\begin{align*}
    \E \{C_n(\btheta)\} = \E_X\Big[&P(Y_i<Y_j,\Delta_i=1)\ind\big\{H(X_i,\bZ_i;\btheta)>H(X_j,\bZ_j;\btheta)\big\}+\\
&P(Y_i>Y_j,\Delta_j=1)\ind\big\{H(X_i,\bZ_i;\btheta)<H(X_j,\bZ_j;\btheta)\big\}\Big],
\end{align*}
which the stochastic ordering shows is maximized at $\btheta_0$. The sample analogue $C_n(\btheta)$ is therefore maximized at $\hat \btheta$. Although \cite{uno2011cstat} shows that the Harrell-type C-index can be biased upward as a measure of discriminatory accuracy in the presence of heavy censoring, this does not affect the consistency of our parameter estimator $\hat\btheta$: the ranking relationships that identify $\btheta_0$ are preserved under conditionally independent censoring, so $\hat\btheta$ converges to $\btheta_0$ regardless of the censoring distribution (see Web Appendix A). Under mild regularity conditions, $\hat \btheta$ uniquely maximizes $C_n(\btheta)$ and is consistent (\citealp{khan2007partial}; see Web Appendix A).

The MCE is a rank-based M-estimator defined through pairwise comparisons, enabling application of the asymptotic theory of \cite{sherman1993} for maximum rank correlation. 
We note that although the asymptotic theory builds on the framework of \cite{sherman1993} for maximum rank correlation estimators, the replacement of the monotone linear index with the non-monotone, non-differentiable function $H(X,\bZ;\btheta)$ requires adaptation of the standard arguments. The piecewise linear structure of the risk index $H(X,\bZ; \btheta)$ and the Euclidean property of its induced function class are exploited in Web Appendix A to handle the non-differentiability at the critical point boundary. These arguments yield the $\sqrt{n}$-consistency and asymptotic normality of the MCE estimator stated in Theorem \ref{thm-1}.

\begin{theorem}\label{thm-1}
Under model (\ref{semi_model}) and assumptions~(\ref{as:G})--(\ref{as:reg}) specified in Web Appendix A, 
$$
\sqrt{n}(\hat\btheta - \btheta_0) \xrightarrow{d} N(0, \bV^{-1}\bU \bV^{-1}).
$$ 
Here, $\bV=2^{-1}\E\big[\boldsymbol{\nabla_2\tau}(\bW_i;\btheta_0)\big]$, $\bU=\E\big[\boldsymbol{\nabla_1}\tau(\bW_i;\btheta_0)\boldsymbol{\nabla_1\tau}(\bW_i;\btheta_0)^{\T}\big]$, where \newline
$\bW_i = (Y_i,\Delta_i, X_i, \bZ_i)^{\T}$, 
and $\nabla_1$ and $\nabla_2$ denote the first- and second-order derivatives of the pairwise estimating function $\tau$ with respect to $\btheta$.
The pairwise estimating function $\tau(\bW_i;\btheta)$ is defined as
\begin{align*}
\begin{split}
\tau(\bW_i,\boldsymbol{\theta})=&\E\Big[\mathds{1} \big\{Y_i<Y,\Delta_i=1\big\}\mathds{1}\big\{H(X_i,\bZ_i;\btheta)>H(X,\bZ;\boldsymbol{\theta})\big\}\Big]\\
& +\E\Big[\mathds{1}\{Y<Y_i,\Delta=1\}\mathds{1}\big\{H(X,\bZ;\boldsymbol{\theta})>H(X_i,\bZ_i;\btheta)\big\}\Big].
\end{split}
\end{align*}
\end{theorem}

The variance components $\bU$ and $\bV$ admit tractable one-dimensional forms by conditioning on the risk index $H(X_i,\bZ_i;\btheta_0)$ rather than on the full covariate vector, following the dimension-reduction argument in \cite{sherman1993}. The explicit one-dimensional integral expressions for $\bU$ and $\bV$, together with the resulting plug-in variance estimator are given in Web Appendix A.

The asymptotic normality established in Theorem \ref{thm-1} provides the basis for inference on $\btheta$. Although the limiting variance has the closed-form representation $\bV^{-1}\bU\bV^{-1}$, estimating $\bU$ and $\bV$ requires numerical differentiation and bandwidth selection; we therefore implement the inference through the bootstrap, which avoids these choices. We approximate the sampling distribution of 
$(\hat{\beta}_0,\log \hat {{\beta}}_1,\hat\balpha_1^{\T}, \hat\balpha_2^{\T})$, where the log transformation maps the constrained parameter $\beta_1 > 0$ to the real line so that the symmetric Wald interval applies. 
Let $\hat{\beta}_{0}^{\ell}$, $\log\hat\beta_{1}^{\ell}$, $\hat\balpha_{1}^{\ell}$ and $\hat\balpha_{2}^{\ell}$, $\ell=1,...,n_{\rm boot}$, denote the estimates from $n_{\rm boot}$ bootstrap samples, and let $SE_{\rm boot}$ denote the standard deviation of bootstrap estimates. The $100(1-\gamma)\%$ confidence intervals for $(\hat{\beta}_0,\log\hat\beta_1,\hat\balpha_1,\hat\balpha_2)$ are
\begin{align*}
&\big(\hat{\beta}_{0}-z_{1-\gamma/2}\,SE_{\rm boot}(\hat{\beta}_{0}),\;\hat{\beta}_{0}+z_{1-\gamma/2}\,SE_{\rm boot}({\hat \beta}_{0})\big),\\
&\big(\log\hat {\beta}_{1}-z_{1-\gamma/2}\,SE_{\rm boot}(\log\hat{\beta}_{1}),\;\log\hat \beta_{1}+z_{1-\gamma/2}\,SE_{\rm boot}(\log\hat {\beta}_{1})\big),\\
&\big(\hat{\balpha}_{1}-z_{1-\gamma/2}\,SE_{\rm boot}(\hat \balpha_1),\;\hat\balpha_1+z_{1-\gamma/2}\,SE_{\rm boot}(\hat  \balpha_1)\big),\\
&\big(\hat\balpha_2-z_{1-\gamma/2}\,SE_{\rm boot}(\hat \balpha_2),\;\hat\balpha_2+z_{1-\gamma/2}\,SE_{\rm boot}(\hat \balpha_2)\big),
\end{align*}
respectively.

\section{Smoothed Kaplan-Meier Estimation}\label{sec:estimation-infinite}
Beyond estimating the parametric components $\btheta$, we propose an estimator for the event risk given biomarker $X_i$ and risk factors $\bZ_i$. Since the nonparametric model components (transformation function $G(\cdot)$ and distribution of error term $\epsilon$) are not individually identifiable, they are identified conjointly in the form of the survival probability conditional on the risk index $H(X_i,\bZ_i;\btheta)$. Specifically, we estimate this conditional survival probability using a smoothed Kaplan-Meier estimation based on kernel weighting. This approach is related to prior work on marker-dependent survival estimation under censoring \citep{altman1992introduction, heagerty2000time}.

Let $(Y_i, \Delta_i)$ denote the observed follow-up time and event indicator, and let $y_{(1)},y_{(2)},\dots, y_{(K)}$ be the ordered uncensored survival times. For a given covariate profile $(X_i, \bZ_i)$, write $\hat H_i=H(X_i, \bZ_i;\hat \btheta)$. For any query value $h_0$, we assign observation-specific weights using a kernel function $K_h(\cdot)$ applied to the difference $h_0-\hat H_j$. Here, $K_h(\hat H_j,\hat H_i)$ denotes a kernel function applied to the difference $\hat H_j-\hat H_i$, so that observations with similar risk index receive large weights while those farther away receive smaller weights. At each survival time $y_{(k)}$, we define the kernel-weighted number of events and the kernel-weighted risk set as the weighted sums over subjects experiencing the event at $y_{(k)}$ and those remaining at risk at $y_{(k)}$, respectively. The smoothed Kaplan-Meier estimator of the conditional survival function is then given by
\begin{align*}
\begin{split}
\widehat{S}_h(t|X_i, \bZ_{i}) = \prod_{y_{(k)}\leq t} \Big\{ 1-\frac{\sum_{j=1}^n K_h(\widehat{H}_j, \widehat{H}_i) \ind(Y_j = y_{(k)}) \Delta_j }{\sum_{j=1}^n K_h(\widehat{H}_j, \widehat{H}_i) \ind(Y_j\geq y_{(k)})} \Big\}.
\end{split}
\end{align*}
Event risk at time $t$ is estimated by $1-\hat S_h(t|X_i,\bZ_i)$. This kernel-weighted formulation smooths the Kaplan-Meier estimator over the risk index space and provides a flexible nonparametric estimate of the conditional survival function given $\hat H_i$.

Kernel choice has minimal effect on the estimator; selection depends on data characteristics. Among common options---nearest-neighbor, Epanechnikov, and Gaussian kernels---nearest-neighbor kernels adaptively set the local neighborhood size and are advantageous where observations are sparse, as at the extremes of the biomarker range, while Gaussian kernels provide smooth, numerically stable weighting. We accordingly use a K-nearest-neighbor kernel in the simulation study, where the biomarker range is sparsely sampled at its ends, and a Gaussian kernel in the real-data analysis.

The bandwidth, by contrast, governs the bias-variance tradeoff and is the main determinant of the estimator's behavior. Smaller bandwidths reduce bias but increase variance, whereas larger bandwidths yield smoother estimates with potentially higher bias. Data-driven procedures such as leave-one-out cross-validation and plug-in methods are available for its selection. We tune the bandwidth in each setting and assess the robustness of the resulting risk estimates through bootstrap resampling.

\section{Critical Point and Critical Region}\label{sec:estimation-critical-point}
 
The critical point identifies the biomarker value of minimal event risk, whereas the critical region identifies the range of biomarker values whose event risk remains close to the minimum.
In this section, we present estimation and inference for these quantities. 

In model (\ref{semi_model}), the critical point is the breakpoint where the two linear components inside the maximum intersect, which is $X_c=(1+\beta_1)^{-1}\cdot(-\beta_0+\bZ^{\T}\balpha_1-\bZ^{\T}\balpha_2).$ Plugging in $\hat\btheta=(\hat\beta_0, \hat\beta_1, \hat\balpha_1, \hat\balpha_2)$, the estimator of the critical point is $\hat X_{c}=(1+\hat \beta_{1})^{-1}\cdot(-\hat \beta_{0}+ \bZ^{\T}\hat\balpha_{1}-\bZ^{\T}\hat \balpha_{2})$. 
Let $g(\hat\btheta)=(1+\hat \beta_{1})^{-1}\cdot(-\hat \beta_{0}+\bZ^{T}\hat\balpha_{1}-\bZ^{T}\hat\balpha_{2})$. By consistency of the MCE estimator $\hat\btheta$ and the continuity of $g$, 
the Delta method gives
$$
    \sqrt{n}(\hat X_{c}-X_c)\rightarrow N\big(0,\nabla g^{\T} V^{-1}U V^{-1}\nabla g\big)
$$
in distribution, where
$$
\nabla g = \Big[-(1+\hat \beta_{1})^{-1},\;
    (1+\hat \beta_{1})^{-2}(\hat \beta_{0}-\bZ^{\T}\hat \balpha_{1}+\bZ^{\T}\hat\balpha_{2}),\;
    (1+\hat \beta_{1})^{-1}\bZ,\;
    -(1+\hat \beta_{1})^{-1}\bZ\Big]^{\T}.
$$ 
The asymptotic variance of $\hat X_c$ is $\text{Var}(\hat X_c) = n^{-1}\nabla g^{\T} V^{-1} U V^{-1} \nabla g$. In practice, we estimate $\text{Var}(\hat X_c)$ via bootstrap, which avoids numerical derivatives and bandwidth selection required for direct estimation of $U$ and $V$. This yields the Wald confidence interval for the critical point, $\hat X_c \pm z_{1-\gamma/2}\sqrt{\widehat{\operatorname{Var}}(\hat X_c)}$.

Given a risk threshold $a$ and a fixed time point $t_0$, the corresponding critical region is the set of biomarker values satisfying
$$
    P(T \leq t_0|X, \bZ) = F_\epsilon\big\{G^{-1}(t_0) + \max(-X+\bZ^{\T}\balpha_1, \beta_0 + \beta_1 X + \bZ^{\T}\balpha_2)\big\} \leq a,
$$
where $F_\epsilon$ is the cumulative distribution function of $\epsilon$. This reduces to finding values of $X$ such that $\max(-X+\bZ^{\T}\balpha_1, \beta_0 + \beta_1 X + \bZ^{\T}\balpha_2) \leq F_\epsilon^{-1}(a) - G^{-1}(t_0) = c$. 
Because the left-hand side is minimized at $X_c$, the critical region is empty if $c<H(X_c,\bZ;\btheta)$. Otherwise, the general form of the critical region is $\big[\bZ^{\T}\balpha_1-c,\;\beta_1^{-1}\cdot(c-\beta_0-\bZ^{\T}\balpha_2)\big]$. 
Although this expression gives the population-level form of the critical region, $G$ and $F_\epsilon$ are not individually identifiable. Therefore, in practice, we estimate the critical region using the smoothed Kaplan--Meier estimator described in Section \ref{sec:estimation-infinite}. Specifically, we estimate the boundaries of the critical region by finding biomarker values at which the estimated event probability $P(T\le t_0\mid X,\bZ)$ reaches $a$.

\section{Simulation Study}\label{sec:simulation}
We evaluate MCE finite-sample performance and compare it with Cox-based alternatives across sample sizes, censoring rates, transformations, and error distributions.

Data are generated from model $T = G\big\{-\max(-X+Z_1\alpha_1, \beta_0 + \beta_1 X + Z_2\alpha_2) + \epsilon\big\}$ with biomarker $X\sim\text{Uniform}(-5,10)$, continuous covariate vector $Z_1\sim N(0,1)$ and binary covariate $Z_2\sim\text{Bernoulli}(0.5)$, intercept $\beta_0=0$, U-shape slopes $\alpha_1=3$ and $\alpha_2=-3$, and right-arm slope $\beta_1\in\{1,2\}$. The error $\epsilon$ follows either $N(0,9)$ or a centered minimum-extreme-value distribution with the same standard deviation. We consider two transformations: the logistic $G(s) = 10\,\text{plogis}(s/5)$ and the exponential $G(s) = \exp\{(s+1)/5\}$. Independent censoring is generated as $C\sim\text{Uniform}(c_l,c_u)$, with $(c_l,c_u)$ tuned per scenario to target a censoring rate of approximately $15\%$, $30\%$, or $50\%$. The main simulation fixes $\beta_1=2$ and varies $n=200,500,1000$, transformation function $G$, and distribution of error term $\epsilon$, giving twenty scenarios (S01--S20); three additional scenarios (S21--S23) fix $\beta_1=1$ with logistic $G$ and min-EV $\epsilon$ to verify performance under a milder U-shape risk pattern. We generate $1{,}000$ replications per scenario.

The error, transformation, and U-shape choices are designed for investigation of distinct robustness features. The normal $N(0,9)$ error is symmetric and a commonly adopted error distribution; the minimum-extreme-value error corresponds to the Cox proportional hazards specification and provides the closest stress test scenario against Cox models that are alternative to our proposed method. The logistic and exponential transformations $G$ realize qualitatively different baseline-hazard shapes. The two right-arm slopes $\beta_1=1,\,2$ separate a pronounced U-shape ($\beta_1=2$) from a milder regime ($\beta_1=1$). Effective sample size across simulation scenarios ranges from $\sim 100$ ($n=200$, $50\%$ censoring) to $\sim 850$ ($n=1000$, $15\%$ censoring).

We obtain MCE point estimates by maximizing the empirical concordance index $C_n(\btheta)$ defined in Section~\ref{sec:estimation-finite} via differential evolution \citep{mullen2011deoptim}, and estimate standard errors via the standard nonparametric bootstrap, and $500$ bootstrap replicates are used per sample. Wald-type $95\%$ confidence intervals are constructed using the bootstrap standard error for the parameters $(\beta_0,\log\beta_1,\alpha_1,\alpha_2)$ and the delta-method standard error for $X_c$. 

We compare the MCE against four feasible Cox models and one infeasible oracle, chosen to span the flexibility-parsimony tradeoff practitioners face when no explicit change-point method is available. The natural-spline Cox model (natural cubic spline with four degrees of freedom) is the most flexible alternative but does not parameterize a change point, so $X_c$ must be recovered indirectly by minimizing the fitted linear predictor over $X$, which is sensitive to boundary noise. The quadratic Cox model (linear plus quadratic biomarker terms) is the minimal-parameter alternative but constrained to symmetric U-shapes. Each is also extended with biomarker--covariate interaction terms, allowing covariate-dependent shifts in the two arms. The oracle piecewise-linear Cox model uses subject-specific critical points $X_c(\bZ_i)$ set at the true values and fits branch-specific slopes, covariate effects, and a right-branch intercept; it is infeasible (true $X_c(\bZ_i)$ unknown in practice) but provides the parametric upper bound for any pre-specified-knot Cox formulation. For each feasible alternative, $X_c$ is the minimizer of the estimated linear predictor over $X$ holding $\bZ=(Z_1,Z_2)=(0.5,1)$ fixed. The MCE is distribution-free in $G$ and $\epsilon$ and correctly specified across all simulation settings. The four feasible Cox alternatives are misspecified in every setting: under minimum-extreme-value error the Cox proportional hazards family is correct but no alternative parameterization contains the true model as a special case; under normal error the Cox proportional hazards family itself is incorrect. The piecewise-linear oracle Cox model contains the true model: under minimum-extreme-value error (with the current $\beta_0 = 0$ intercept convention) it is correctly specified, while under normal error it remains the best Cox approximation given correct branch assignment.

Table~\ref{tab:sim-main} reports bias, empirical standard error (ESE), mean of estimated standard errors (ASE), and empirical $95\%$ coverage probability (ECP) for $(\beta_0,\log\beta_1,\alpha_1,\alpha_2,X_c)$ across all 23 simulation scenarios. Across all scenarios the MCE estimator is nearly unbiased: the bias of $\hat X_c$ is below $0.10$ at $n=200$ and below $0.01$ at $n=1000$, and the bias of each $(\beta_0,\log\beta_1,\alpha_1,\alpha_2)$ is at most one-third of its empirical standard error. The empirical standard errors scale as $n^{-1/2}$, consistent with the asymptotic theory of Section~\ref{sec:estimation-finite}. At $n=500$ and $n=1000$, the bootstrap standard errors closely track or modestly exceed the empirical standard errors, and Wald $95\%$ confidence intervals attain coverage at or slightly above the nominal level for all four parameters and $X_c$, reflecting the mild conservatism characteristic of nonparametric bootstrap inference for rank-based objective functions. At $n=200$, coverage falls modestly below nominal, particularly for $\log\beta_1$ ($\text{ECP}\approx 0.83\text{--}0.86$) and $X_c$ ($\text{ECP}\approx 0.84\text{--}0.87$);  we view this as a finite-sample limitation when the effective sample size (number of observed events) is small. For practical inference, we recommend an effective sample size of at least $350$ events, corresponding to total $n\geq 500$ at $\approx 30\%$ censoring or $n\geq 700$ at $\approx 50\%$ censoring. Performance improves uniformly with sample size across all parameters and critical-point quantities.

Figure~\ref{fig:sim-xc} and Table~\ref{tab:sim-cindex} (Web Appendix C) report critical-point estimation accuracy across the four feasible Cox competitors, the MCE, and the oracle benchmark; Figure~\ref{fig:sim-cindex} and Table~\ref{tab:supp-cindex-all} (Web Appendix C) summarize the corresponding test-set Harrell's $C$-index. The MCE recovers $X_c$ with negligible bias ($|\text{Bias}|\leq 0.10$ at $n=200$ and $\leq 0.01$ at $n=1000$) and coverage at or near the nominal $95\%$ level. The four Cox alternatives, in contrast, exhibit systematic bias of $-0.5$ to $-2.6$ that does not diminish with $n$; their nominal $95\%$ confidence intervals attain below-nominal coverage in every scenario, falling to essentially $0\%$ for the additive specifications. This pattern follows directly from the compatibility analysis above: because none of the Cox alternatives' parameterizations contain $H$ as a special case, the $X_c$ recovered by minimizing each fitted linear predictor corresponds to the minimum of the best-approximating Cox surface rather than the true critical point, and this projection error persists at any sample size. The MCE, by parameterizing $X_c$ directly as a function of $(\beta_0, \log\beta_1, \alpha_1, \alpha_2)$, both estimates the critical point consistently and produces bootstrap-based delta-method confidence intervals that achieve nominal coverage.

The global test-set Harrell's $C$-index tells a more compressed story: against the four feasible Cox competitors---none of which knows the true critical point---the MCE achieves higher $C$-index by $2$--$3\%$ uniformly across all scenarios. Such consistent gains are noteworthy given that the $C$-index is known to be relatively insensitive to differences between risk models \citep{cook2007roc}. The oracle Cox serves as an unrealistic upper bound that uses each subject's true $X_c(\bZ)$: under the minimum-extreme-value error specification (where Cox proportional hazards is correct), the oracle matches or marginally exceeds the MCE; under the normal error specification (where Cox proportional hazards is the wrong family), the MCE matches or exceeds even the oracle. Together, the contrast between the compressed $C$-index gap and the dramatic $X_c$ gap underscores that the methods differ qualitatively in what they estimate: the explicit U-shape parameterization of the MCE recovers the change-point structure that alternative parameterizations only approximate, without requiring any knowledge of the truth. Figures~\ref{fig:sim-xc} and~\ref{fig:sim-cindex} provide a visual summary of both findings across all $23$ scenarios. In Figure~\ref{fig:sim-xc}, panel~(a) (main $(G,\epsilon)$ grid at $\beta_1=2$, $30\%$ censoring) shows the systematic leftward bias of all four Cox alternatives versus the MCE's near-zero centering, while panels~(b) and~(c) confirm that the same separation persists under heavier or lighter censoring and under the milder $\beta_1=1$ U-shape. Figure~\ref{fig:sim-cindex} displays the corresponding compressed $C$-index gap, together with the MCE's parity with---or marginal advantage over---the oracle, in matching three-panel layout.

These results confirm that the MCE provides unbiased point estimates, asymptotically valid Wald inference for all model parameters---including the critical point $X_c$, the central estimand of this paper---and improved predictive $C$-index relative to commonly used Cox extensions, across the range of $G$, $\epsilon$, and censoring conditions tested. At the smallest sample size studied, valid inference on $\log\beta_1$ becomes fragile and benefits from a moderately larger sample; inference for the critical point $X_c$ attains nominal coverage at moderate sample sizes (effective sample size of roughly $350$ events or more) across every scenario.

\section{Real Data Analyses on UK Biobank Data}\label{sec:data-analyses}

We apply the MCE to UK Biobank data to estimate the association between baseline BMI and all-cause mortality, adjusted for gender and age. We fit the model
$$
T = G\big\{-\max(-X+\bZ^{\T}\balpha_1, \beta_0 + \beta_1 X + \bZ^{\T}\balpha_2) + \epsilon\big\},
$$
where $T$ is the time from enrollment to death or censoring, $X$ is baseline BMI, $\bZ = (Z_1, Z_2)^{\T}$ denotes the covariates including sex and age group.

Initial values for MCE optimization are obtained by exploiting the U-shaped structure of $H(X,\bZ; \btheta)$: at any subgroup-specific risk level, the two biomarker roots lying on opposite branches satisfy a linear relationship whose slope and intercept identify $\beta_0$ and $\beta_1$. Empirical paired roots are extracted from LOESS-smoothed incidence rates on a binned biomarker grid, then refined via COBYLA optimization with jitter-based robustness checks; full derivations and algorithmic details are given in Web Appendix B.

The UK Biobank is a large-scale, population-based prospective cohort that recruited 502,188 participants aged 38-73 years (at the time of recruitment) from across the United Kingdom between 2006 and 2010. All volunteers gave electronic informed consent, completed touchscreen questionnaires, and underwent standardized physical examinations. Baseline data were subsequently linked to National Health Service Hospital Episode Statistics, cancer registries, and national death records, providing passive, longitudinal follow‑up of morbidity and mortality.

We merged the participant data with death registry records using participant IDs. Dates of death were converted to calendar dates, and follow-up was censored at December 31, 2023, corresponding to the end of the death registry linkage period.
Participants were excluded if they were missing identifiers, lacked BMI records, or had extreme BMI values outside the range of 15–40 kg/m$^2$. After applying these criteria, 486,944 remained, of whom 42,076 experienced an event (defined as death after baseline) and 444,868 were censored.
The censoring date was taken as the earliest of lost-to-follow-up or the registry end date. BMI at baseline was denoted by $X_i$.
Covariates were defined as: $\bZ=(Z_1,Z_2)$, where $Z_1=1$ if male and 0 if female, and $Z_2=1$ if age at recruitment was greater than 65 years and 0 otherwise.
To reduce computational burden while preserving subgroup balance, we randomly selected 10,000 participants for model fitting. 
The estimated parameters of our model are shown in Table \ref{tab:est_result1}.

To estimate the survival probability function, we apply kernel-weighted Kaplan--Meier smoothing followed by isotonic regression via the pool-adjacent-violator algorithm (PAVA) \citep{mammen1991estimating}. Although $S(t;X,\bZ)$ is monotone in $H(X,\bZ;\btheta)$ by construction, finite-sample variation and kernel smoothing can introduce small violations; PAVA applied to $\hat S(t\mid H_{(i)})$ sorted by increasing $H_{(i)}$ restores monotonicity, after which $\tilde{R}(t|H)=1-\tilde{S}(t|H)$ yields smooth and theoretically consistent U-shaped risk curves with respect to $H(X,\bZ;\hat\btheta)=\max(-X+\bZ^{\T}\hat\balpha_1, \hat\beta_0 + \hat\beta_1 X + \bZ^{\T}\hat\balpha_2)$, where $X$ denotes BMI and $\bZ$ represents subgroup indicators.

As shown in Figure~\ref{fig:fix_time}, the estimated cumulative risk varies nonlinearly with BMI at both 10 and 15 years of follow-up, with elevated risk at both low and high BMI levels and a well-defined minimum corresponding to the subgroup-specific critical point. The location of this minimum differs systematically by age: individuals younger than 65 tend to have lower critical point values than those aged 65 or older (Table~\ref{tab:est_result1}), suggesting that the optimal BMI range may vary across life stages. The estimated risk curves also enable identification of the critical region for any given risk threshold by inverting the curve, as described in Section \ref{sec:estimation-critical-point}. For females under 65 at 15 years of follow-up, setting a risk threshold of $a=0.05$, the estimated critical region is (20.5, 30.0) kg/m$^2$, indicating the BMI range associated with a 15-year all-cause mortality risk below 5\%. The 95\% bootstrap confidence intervals for the lower and upper boundaries of this region are $(19.6, 21.0)$ and $(28.5, 31.1)$ kg/m$^2$, respectively.

Figure \ref{fig:fix_bmi} demonstrates that cumulative risk increases monotonically over time at fixed BMI levels. Together, Figures~\ref{fig:fix_time}--\ref{fig:fix_bmi} highlight two complementary aspects of the model: (i) a consistent subgroup-specific U-shaped relationship between BMI and risk at fixed time points, and (ii) steadily increasing cumulative risk trajectories over time at fixed BMI. The empirical patterns align closely with the structural form of the fitted risk index $H(X,\bZ;\hat\btheta)$, providing indirect support for the proposed semiparametric modeling framework.

\section{Discussion}\label{sec:discussions}

To characterize U-shaped biomarker--survival relationships, we developed a semiparametric model that explicitly parameterizes the critical point and proposed a maximum C-index estimator paired with a smoothed Kaplan--Meier estimator for the unspecified transformation. This framework extends change-point methodology to censored survival data and yields direct estimation and inference for the critical point and region---quantities that ad-hoc change-point smoothing methods cannot accommodate.

Our framework has four methodological strengths. (i) Rank-based estimation is scale-invariant and robust to outliers, so results are unaffected by biomarker rescaling. (ii) The semiparametric structure is flexible: the critical point and region are explicitly parameterized while the transformation is unspecified beyond monotonicity, accommodating a wide range of U-shaped patterns without prespecifying the critical point. (iii) Empirical-process theory yields asymptotic properties of the estimators and valid confidence intervals for finite-dimensional parameters and the critical point. (iv) A well-designed initialization strategy circumvents global search for computational efficiency. Simulations confirmed good finite-sample performance: across error distributions, transformations, and censoring rates, the MCE exhibited small bias and MSE in estimating both the parameters and the critical point, with both decreasing as sample size grew.

Our real-data analyses using the UK Biobank cohort provide substantive insights into the association between BMI and all-cause mortality. Across sex and age subgroups, we observed a consistent U-shaped relationship between BMI and cumulative mortality risk, with subgroup-specific differences in the location of the critical point corresponding to the minimum risk.
In particular, the estimated critical point differed systematically by age and sex, with younger and female individuals tending to exhibit lower critical points than older individuals. While the overall U-shaped pattern was stable across subgroups, the width of the confidence intervals for the estimated critical points varied, reflecting differing levels of uncertainty in the identification of an “optimal” BMI across populations.
These findings suggest that although BMI serves as a reproducible population-level risk stratification marker, its precision for defining individualized risk thresholds is limited and heterogeneous across subgroups. Importantly, by providing confidence intervals for the estimated critical points, our method offers a principled way to quantify uncertainty in subgroup-specific optimal BMI values, thereby informing cautious interpretation and potential clinical use.

Although the framework addresses an important class of problems, several extensions merit further development. First, the estimator depends on well-chosen initial values; further automation of initialization could enhance robustness. Second, the current model identifies a single critical point for a single biomarker; natural generalizations include multiple turning points and multivariate biomarkers (via a general tree structure), longitudinal or time-varying biomarkers (for landmark prediction), and competing risks. Finally, the framework assumes a U-shaped risk relationship; formal diagnostics for assessing this assumption and procedures for model comparison among competing functional forms are needed. The method is most applicable when a U-shaped association is suspected and the scientific interest is the critical point; exploratory Kaplan--Meier curves by biomarker quantiles can help assess plausibility before applying the framework.

\backmatter

\section*{Acknowledgements}

This work was supported by the National Institute on Aging (R03AG083470). The content is solely the responsibility of the authors and does not necessarily represent the official views of the National Institutes of Health. The authors are grateful for Megan Clark's assistance in help revising the manuscript.

\textbf{AI-use disclosure.} In accordance with the Journal's policy and the COPE position statement on Authorship and AI, we disclose that a large-language-model assistant (Anthropic Claude, accessed through the Claude Code CLI) was used during the preparation of this manuscript for language-level editing (copyediting, tightening prose, and improving clarity), and LaTeX formatting of tables and figures that the authors subsequently reviewed and revised. All statistical methods, theoretical derivations, simulation designs, data analyses, and scientific conclusions were developed, verified, and are the sole responsibility of the authors. The AI tool did not generate or independently verify any citations, statistical results, or numerical outputs in the manuscript. No AI tool is listed as an author, consistent with the requirement that authorship entails accountability that AI systems cannot bear.
\vspace*{-8pt}

\section*{Conflict of Interest}
The authors declare no conflicts of interest.

\section*{Supplementary Materials}
Web Appendix A, referenced in Sections \ref{sec:estimation-finite} and \ref{sec:estimation-critical-point}, contains regularity assumptions, proofs of consistency and asymptotic normality of the maximum C-index estimator, and the U-statistic decomposition. The supplementary materials are included with this preprint.

\section*{Data Availability Statement}
The data used in this study are available from the UK Biobank (\url{https://www.ukbiobank.ac.uk/}) upon application and approval. Due to data access restrictions, the data cannot be shared publicly. Code used to implement the proposed method and reproduce the analyses is available at \url{https://github.com/Rita-zrFu/Ushape}.

\bibliographystyle{biom} \bibliography{biomsample}

\clearpage
\begin{sidewaystable}[hb]
\centering
\caption{MCE finite-sample performance across all simulation scenarios. }\label{tab:sim-main}
\vspace{-6pt}
\footnotesize
\setlength{\tabcolsep}{3pt}
\renewcommand{\arraystretch}{0.95}
\resizebox{\textwidth}{!}{%
\begin{tabular}{ccccc ccccc ccccc ccccc ccccc ccccc}
\hline
& & & & & \multicolumn{4}{c}{$\beta_0$} & & \multicolumn{4}{c}{$\log\beta_1$} & & \multicolumn{4}{c}{$\alpha_1$} & & \multicolumn{4}{c}{$\alpha_2$} & & \multicolumn{4}{c}{$X_c$} \\
\cline{6-9}\cline{11-14}\cline{16-19}\cline{21-24}\cline{26-29}
$n$ & $G$ & $\epsilon$ & $\beta_1$ & Cens & Bias & ESE & ASE & ECP & & Bias & ESE & ASE & ECP & & Bias & ESE & ASE & ECP & & Bias & ESE & ASE & ECP & & Bias & ESE & ASE & ECP \\
\hline
$200$ & logistic & $N(0,9)$ & $2$ & 15\% & $-0.39$ & $1.57$ & $1.40$ & $0.93$ & & $+0.08$ & $0.33$ & $0.31$ & $0.88$ & & $+0.24$ & $1.01$ & $0.96$ & $0.95$ & & $-0.43$ & $1.50$ & $1.31$ & $0.92$ & & $+0.06$ & $0.50$ & $0.47$ & $0.88$ \\
$200$ & logistic & $N(0,9)$ & $2$ & 30\% & $-0.52$ & $1.75$ & $1.49$ & $0.89$ & & $+0.09$ & $0.35$ & $0.31$ & $0.86$ & & $+0.27$ & $1.09$ & $1.03$ & $0.94$ & & $-0.50$ & $1.59$ & $1.35$ & $0.91$ & & $+0.08$ & $0.55$ & $0.49$ & $0.86$ \\
$200$ & logistic & $N(0,9)$ & $2$ & 50\% & $-0.59$ & $2.12$ & $1.60$ & $0.84$ & & $+0.10$ & $0.41$ & $0.33$ & $0.81$ & & $+0.35$ & $1.36$ & $1.14$ & $0.92$ & & $-0.61$ & $1.73$ & $1.41$ & $0.89$ & & $+0.08$ & $0.68$ & $0.53$ & $0.82$ \\
$500$ & logistic & $N(0,9)$ & $2$ & 15\% & $-0.19$ & $0.94$ & $1.01$ & $0.96$ & & $+0.04$ & $0.21$ & $0.22$ & $0.95$ & & $+0.09$ & $0.56$ & $0.63$ & $0.97$ & & $-0.17$ & $0.84$ & $0.92$ & $0.98$ & & $+0.03$ & $0.31$ & $0.36$ & $0.95$ \\
$500$ & logistic & $N(0,9)$ & $2$ & 30\% & $-0.27$ & $1.13$ & $1.12$ & $0.95$ & & $+0.05$ & $0.24$ & $0.24$ & $0.94$ & & $+0.12$ & $0.61$ & $0.70$ & $0.96$ & & $-0.23$ & $0.93$ & $0.99$ & $0.97$ & & $+0.05$ & $0.36$ & $0.39$ & $0.94$ \\
$500$ & logistic & $N(0,9)$ & $2$ & 50\% & $-0.37$ & $1.50$ & $1.30$ & $0.92$ & & $+0.06$ & $0.30$ & $0.26$ & $0.89$ & & $+0.15$ & $0.81$ & $0.82$ & $0.96$ & & $-0.32$ & $1.21$ & $1.08$ & $0.94$ & & $+0.05$ & $0.48$ & $0.45$ & $0.89$ \\
$1000$ & logistic & $N(0,9)$ & $2$ & 30\% & $-0.08$ & $0.67$ & $0.78$ & $0.97$ & & $+0.01$ & $0.15$ & $0.17$ & $0.97$ & & $+0.03$ & $0.40$ & $0.47$ & $0.97$ & & $-0.07$ & $0.60$ & $0.68$ & $0.97$ & & $+0.01$ & $0.25$ & $0.30$ & $0.96$ \\
$200$ & logistic & min-EV & $2$ & 30\% & $-0.60$ & $1.88$ & $1.50$ & $0.87$ & & $+0.11$ & $0.38$ & $0.31$ & $0.83$ & & $+0.31$ & $1.21$ & $1.05$ & $0.93$ & & $-0.59$ & $1.63$ & $1.35$ & $0.91$ & & $+0.10$ & $0.60$ & $0.50$ & $0.84$ \\
$500$ & logistic & min-EV & $2$ & 30\% & $-0.19$ & $1.11$ & $1.10$ & $0.96$ & & $+0.03$ & $0.24$ & $0.24$ & $0.94$ & & $+0.08$ & $0.62$ & $0.70$ & $0.98$ & & $-0.21$ & $0.94$ & $0.96$ & $0.96$ & & $+0.03$ & $0.37$ & $0.40$ & $0.92$ \\
$1000$ & logistic & min-EV & $2$ & 30\% & $-0.07$ & $0.70$ & $0.79$ & $0.98$ & & $+0.01$ & $0.15$ & $0.17$ & $0.97$ & & $+0.03$ & $0.42$ & $0.47$ & $0.96$ & & $-0.06$ & $0.57$ & $0.66$ & $0.97$ & & $+0.00$ & $0.24$ & $0.30$ & $0.98$ \\
$200$ & exp & $N(0,9)$ & $2$ & 15\% & $-0.40$ & $1.55$ & $1.37$ & $0.92$ & & $+0.08$ & $0.33$ & $0.30$ & $0.88$ & & $+0.25$ & $1.05$ & $0.96$ & $0.95$ & & $-0.42$ & $1.47$ & $1.30$ & $0.94$ & & $+0.06$ & $0.50$ & $0.46$ & $0.89$ \\
$200$ & exp & $N(0,9)$ & $2$ & 30\% & $-0.41$ & $1.67$ & $1.45$ & $0.91$ & & $+0.07$ & $0.34$ & $0.31$ & $0.86$ & & $+0.21$ & $1.04$ & $1.00$ & $0.95$ & & $-0.42$ & $1.54$ & $1.32$ & $0.92$ & & $+0.06$ & $0.54$ & $0.49$ & $0.88$ \\
$200$ & exp & $N(0,9)$ & $2$ & 50\% & $-0.61$ & $2.15$ & $1.58$ & $0.83$ & & $+0.11$ & $0.41$ & $0.31$ & $0.79$ & & $+0.33$ & $1.31$ & $1.12$ & $0.92$ & & $-0.62$ & $1.76$ & $1.38$ & $0.87$ & & $+0.08$ & $0.68$ & $0.52$ & $0.82$ \\
$500$ & exp & $N(0,9)$ & $2$ & 15\% & $-0.18$ & $0.93$ & $1.01$ & $0.97$ & & $+0.03$ & $0.21$ & $0.22$ & $0.96$ & & $+0.10$ & $0.55$ & $0.63$ & $0.97$ & & $-0.16$ & $0.82$ & $0.93$ & $0.98$ & & $+0.03$ & $0.31$ & $0.36$ & $0.96$ \\
$500$ & exp & $N(0,9)$ & $2$ & 30\% & $-0.27$ & $1.12$ & $1.10$ & $0.95$ & & $+0.05$ & $0.23$ & $0.24$ & $0.94$ & & $+0.11$ & $0.60$ & $0.68$ & $0.97$ & & $-0.22$ & $0.92$ & $0.98$ & $0.97$ & & $+0.05$ & $0.35$ & $0.39$ & $0.93$ \\
$500$ & exp & $N(0,9)$ & $2$ & 50\% & $-0.36$ & $1.49$ & $1.30$ & $0.91$ & & $+0.05$ & $0.30$ & $0.26$ & $0.90$ & & $+0.14$ & $0.81$ & $0.81$ & $0.95$ & & $-0.29$ & $1.15$ & $1.07$ & $0.94$ & & $+0.05$ & $0.48$ & $0.45$ & $0.89$ \\
$1000$ & exp & $N(0,9)$ & $2$ & 30\% & $-0.08$ & $0.65$ & $0.76$ & $0.97$ & & $+0.01$ & $0.15$ & $0.17$ & $0.96$ & & $+0.04$ & $0.40$ & $0.45$ & $0.97$ & & $-0.07$ & $0.58$ & $0.67$ & $0.97$ & & $+0.01$ & $0.24$ & $0.29$ & $0.96$ \\
$200$ & exp & min-EV & $2$ & 30\% & $-0.60$ & $1.91$ & $1.47$ & $0.86$ & & $+0.10$ & $0.38$ & $0.30$ & $0.81$ & & $+0.31$ & $1.15$ & $1.03$ & $0.94$ & & $-0.56$ & $1.61$ & $1.32$ & $0.90$ & & $+0.09$ & $0.60$ & $0.49$ & $0.83$ \\
$500$ & exp & min-EV & $2$ & 30\% & $-0.21$ & $1.12$ & $1.10$ & $0.96$ & & $+0.04$ & $0.23$ & $0.23$ & $0.93$ & & $+0.09$ & $0.62$ & $0.69$ & $0.97$ & & $-0.22$ & $0.95$ & $0.95$ & $0.96$ & & $+0.03$ & $0.37$ & $0.40$ & $0.93$ \\
$1000$ & exp & min-EV & $2$ & 30\% & $-0.06$ & $0.71$ & $0.79$ & $0.98$ & & $+0.01$ & $0.15$ & $0.17$ & $0.97$ & & $+0.04$ & $0.41$ & $0.46$ & $0.96$ & & $-0.06$ & $0.57$ & $0.66$ & $0.98$ & & $+0.00$ & $0.25$ & $0.30$ & $0.97$ \\
$200$ & logistic & min-EV & $1$ & 30\% & $-0.29$ & $1.53$ & $1.33$ & $0.91$ & & $+0.04$ & $0.34$ & $0.34$ & $0.92$ & & $+0.21$ & $0.92$ & $0.88$ & $0.95$ & & $-0.26$ & $1.30$ & $1.22$ & $0.94$ & & $+0.08$ & $0.80$ & $0.74$ & $0.89$ \\
$500$ & logistic & min-EV & $1$ & 30\% & $-0.06$ & $0.82$ & $0.89$ & $0.96$ & & $+0.00$ & $0.20$ & $0.21$ & $0.96$ & & $+0.05$ & $0.47$ & $0.55$ & $0.97$ & & $-0.09$ & $0.73$ & $0.81$ & $0.97$ & & $+0.02$ & $0.48$ & $0.54$ & $0.97$ \\
$1000$ & logistic & min-EV & $1$ & 30\% & $-0.01$ & $0.54$ & $0.60$ & $0.96$ & & $-0.00$ & $0.13$ & $0.14$ & $0.97$ & & $+0.02$ & $0.32$ & $0.36$ & $0.96$ & & $-0.01$ & $0.49$ & $0.54$ & $0.96$ & & $-0.01$ & $0.33$ & $0.37$ & $0.97$ \\
\hline
\end{tabular}}

\vspace{0.4em}
\begin{minipage}{\textwidth}\footnotesize\raggedright
Bias, empirical standard error (ESE), mean estimated standard error (ASE), and empirical $95\%$ coverage probability (ECP) for $(\beta_0,\log\beta_1,\alpha_1,\alpha_2)$ and the critical point $X_c$, by sample size $n$, transformation $G$, error distribution $\epsilon$, right-arm slope $\beta_1$, and censoring rate. Results are based on $1{,}000$ replications per scenario. ECP is the empirical $95\%$ coverage probability of the bootstrap-based Wald confidence interval; for $X_c$, the standard error is computed via the delta method using the bootstrap covariance of $\hat\btheta$. Scenarios S01--S23 span sample sizes $n\in\{200,500,1000\}$, transformation $G\in\{\text{logistic},\exp\}$, error distribution $\epsilon\in\{N(0,9),\text{min-EV}\}$, right-arm slopes $\beta_1\in\{1,2\}$, and censoring rates $\in\{15\%, 30\%, 50\%\}$.
\end{minipage}
\end{sidewaystable}
\clearpage
\begin{sidewaystable}[htbp]
\centering
\caption{Estimated model parameters and critical point with 95\% confidence intervals}
\label{tab:est_result1}

\renewcommand{\arraystretch}{1.2}
\setlength{\tabcolsep}{6pt}

\begin{tabular}{lcccc}
\hline
Group
& $\hat\balpha_2$ (95\% CI)
& $\hat\beta_0$ (95\% CI)
& $\hat\beta_1$ (95\% CI)
& Critical point (95\% CI) \\
\hline
Female, $<65$ 
& /& $3.26\;(2.60,\;4.15)$& $0.25\;(0.20,\;0.47)$& $22.39\;(21.71,\;23.02)$ \\

Female, $\geq 65$ 
& $11.10\;(8.77,\;13.45)$ 
& $12.55\;(10.25,\;14.63)$& $1.12\;(0.84,\;2.13)$& $24.32\;(22.66,\;25.72)$ \\

Male, $<65$ 
& $4.70\;(2.26,\;5.72)$ 
& $6.54\;(4.55,\;7.43)$& $0.48\;(0.36,\;0.52)$& $23.75\;(23.13,\;24.41)$ \\

Male, $\geq 65$ 
& $24.28\;(21.34,\;39.68)$ 
& $19.23\;(17.65,\;36.48)$& $0.63\;(0.59,\;1.17)$& $28.10\;(18.30,\;32.57)$ \\
\hline
\end{tabular}

\vspace{1cm}

\begin{minipage}{\textwidth}
\footnotesize
\raggedright
Notes: Parameter estimates are obtained from the proposed kernel-weighted Kaplan--Meier model with isotonic regression.
Confidence intervals are based on bootstrap resampling.
Female, $<65$ is the reference subgroup ($Z=(0,0)$), for which $\hat \alpha_2=0$ by construction. The displayed $\hat \alpha_2$ is a subgroup-specific scalar (denoted $\alpha_{2z}$ in Web Appendix B).
$\hat\balpha_1$ is fixed to 0 for identifiability (see Web Appendix B).

\end{minipage}

\end{sidewaystable}
\clearpage
\clearpage
\begin{figure}[!htbp]
    \centering
    \includegraphics[width=1\linewidth,height=0.55\textheight,keepaspectratio]{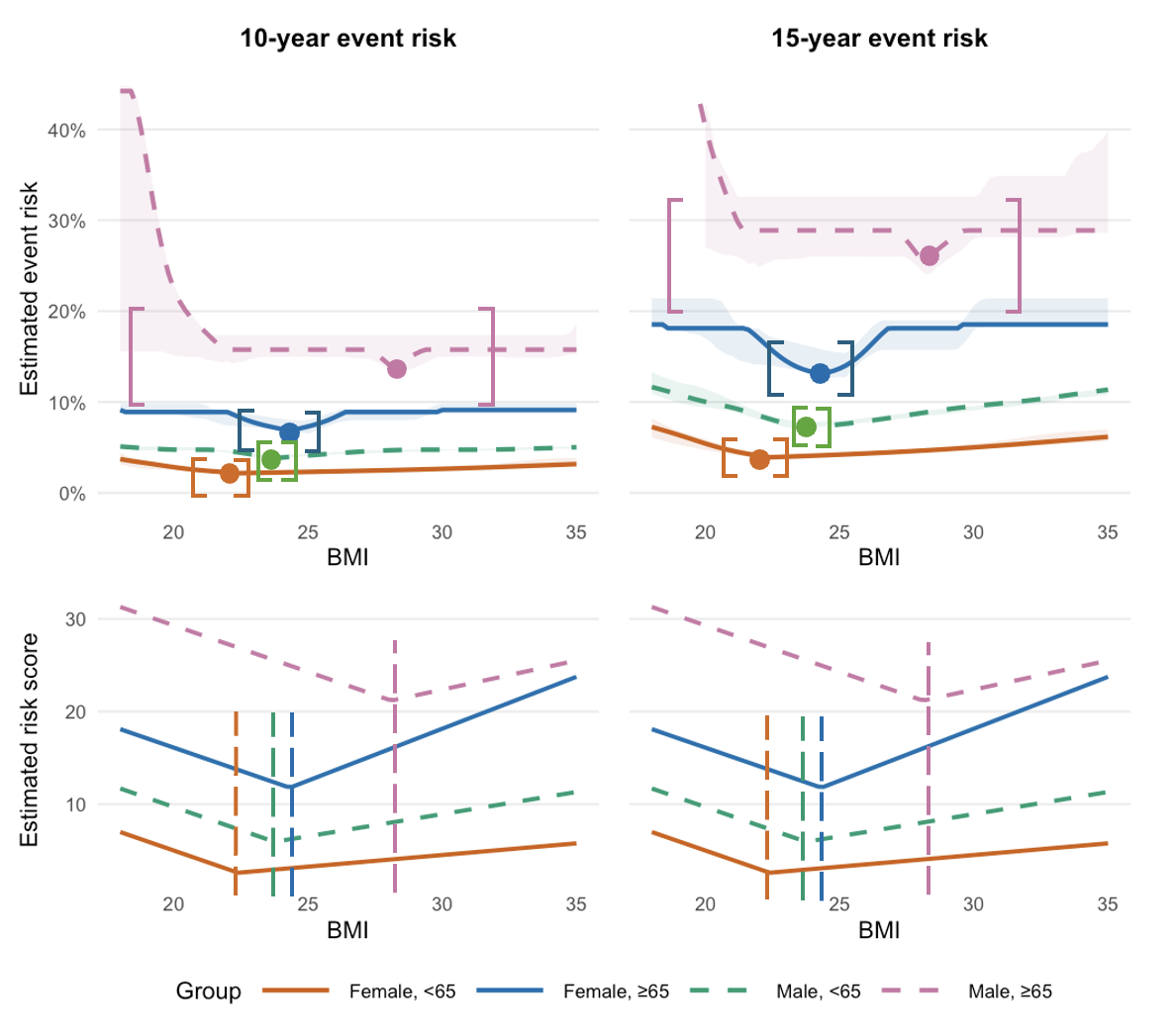}
    \caption{BMI--risk relationship and estimated critical points across sex and age subgroups. The upper panels show the estimated cumulative risk $1-S(t)$ as a function of body mass index (BMI) at 10 and 15 years of follow-up (left and right panels, respectively), stratified by sex and age ($<65$ years vs $\geq 65$). Risk curves were obtained using a kernel-weighted Kaplan--Meier estimator combined with isotonic regression (PAVA) to enforce monotonicity. Shaded ribbons represent 95\% bootstrap confidence intervals (200 resamples). Both panels share the same y-axis scale to facilitate direct comparison across time horizons. Colored rectangles indicate the 95\% confidence intervals of the estimated critical point for each subgroup. The lower panels display the estimated risk index $H(BMI)$ for each subgroup, with vertical dashed lines marking the estimated critical point. Together, the two rows illustrate how subgroup-specific critical points in the risk index translate into distinct BMI--risk profiles over different follow-up times.}
    \label{fig:fix_time}
\end{figure}
\clearpage
\clearpage
\begin{figure}[!htbp]
    \centering
    \includegraphics[width=1\linewidth]{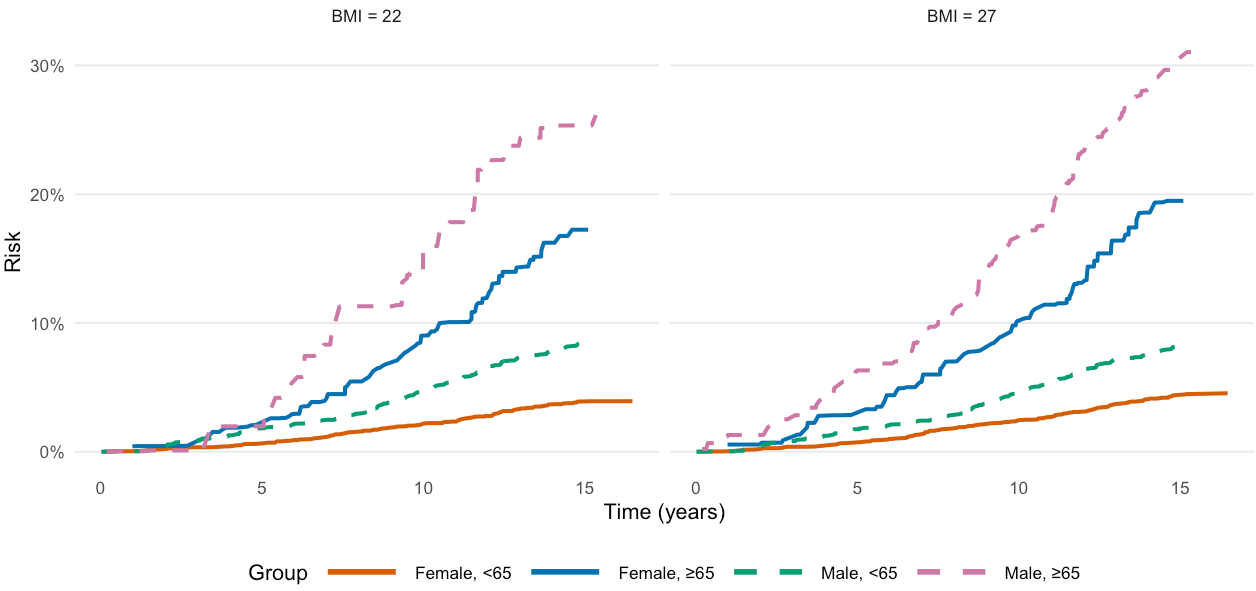}
    \caption{Estimated cumulative risk over time at fixed BMI levels for each subgroup. The $x$-axis represents follow-up time (years), and the $y$-axis represents cumulative risk. Each curve corresponds to a subgroup defined by sex and age ($<65$ vs $\geq 65$). Results are shown at BMI = 22 and BMI = 27. The monotonic increase in risk over time is consistent across all subgroups.}
    \label{fig:fix_bmi}
\end{figure}

\label{lastpage}

\clearpage
\renewcommand{\thesubsection}{A.\arabic{subsection}}
\setcounter{figure}{0}
\setcounter{table}{0}
\renewcommand{\thefigure}{S\arabic{figure}}
\renewcommand{\thetable}{S\arabic{table}}

\begin{center}
{\Large\textbf{Supplementary Material for ``Rank-Based Estimation of U-Shaped Biomarker Risk Curves and Critical Points for Time-to-Event Outcomes''}}\\[1em]
{Zhirui Fu, Mei-Cheng Wang, Yu Du, and Yuxin Zhu}
\end{center}

\section*{Web Appendix A: Proofs of Theorem \ref{thm-1}}

In this section, we state the assumptions and prove Theorem \ref{thm-1}. We first give assumptions and establish the consistency of the the maximum C-index estimator (MCE). We then state two lemmas that establish the $\sqrt{n}$-consistency and asymptotic normality of a general extremum estimator, 
together with a $U$-statistic decomposition and a uniform bound for degenerate $U$-processes; these are general results, due to \citet{sherman1993}, and are reproduced here with proofs for completeness. The remaining sections apply them to the MCE, verify the Euclidean property of the function class induced by $H$,
and give one-dimensional forms of the variance components $U$ and $V$.

\subsection{Assumptions}\label{sec:assump}

\noindent Assumptions under model (\ref{semi_model}) of the main text to establish the consistency of the MCE:

    \assum{as:G} $G$ is a strictly increasing function;
    \assum{as:eps} $\{\epsilon_i\}$ are iid random variables with a continuous distribution supported on $\mathbb{R}$;
    \assum{as:cens} $\{(x_i,\bz_i)\}$ are iid $(1+p)$-vector random variables, independent of $\{\epsilon_i\}$, with distribution function $F_w$. The censoring times satisfy $C_i\perp T_i\mid (x_i,\bz_i)$ and $\Pr (T_i\leq C_i\mid x_i,\bz_i)>0$ for almost every $(x,\bz)$ in the support of $F_w$; 
    \assum{as:dens} $x$ has a distribution with everywhere positive Lebesgue density, conditional on $z_i$;
    \assum{as:rank} $\mathbb{E}[(1,\bz_i^{\T})^{\T}(1,\bz_i^{\T})]$ is nonsingular;
    \assum{as:theta} $\Theta$ is a compact subset of $\mathbb{R}^{2+2p}$ containing $\btheta_0$, on which $\beta_1$ is bounded away from 0.

Assumptions (\ref{as:eps}) and (\ref{as:dens}) ensure that ties in the index occur with probability zero. Assumption (\ref{as:cens}) requires that censoring be conditionally independent of the failure time given the covariates and that the probability of censoring is not equal to one for any $(x,\mathbf z)$; it delivers the pairwise inequality of Web Appendix \ref{sec:consistency}. Assumption (\ref{as:dens}) provides a sufficient condition for point identification of both branches of $H$, and hence of the critical point. Assumption (\ref{as:rank}) rules out additive shifts of the index and guarantees that $\btheta_0$ is identified. Assumption (\ref{as:theta}) provides the compactness needed for uniform convergence, and $\beta_1>0$ fixes the orientation of the U shape. Assumptions (\ref{as:dens})–(\ref{as:theta}) guarantee that the objective function $C_n(\beta_0,\beta_1,\balpha_1,\balpha_2)$ is uniquely maximized at the truth.

\subsection{Proof of Consistency}\label{sec:consistency}

Denote $\bw_i=(x_i,\bz_i)$, with the support $S_w$, $H_i(\btheta)=H(x_i,\bz_i;\btheta)=\max (-x_i+\bz_i^{\T}\balpha_1,\beta_0+\beta_1 x_i+\bz_i^{\T}\balpha_2)$, $\btheta=(\beta_0,\beta_1,\balpha_1^{\T},\balpha_2^{\T})^{\T}$. 
$\btheta_0=(\beta_0^{(0)},\beta_1^{(0)},\balpha_1^{(0){\T}},\balpha_2^{(0)\T})^{\T}$ is the true value.
Under model (\ref{semi_model}) of the main text and assumptions (\ref{as:G}) - (\ref{as:theta}), a crucial step in showing consistency is identification. We first establish the following pairwise inequality.

\begin{lemma}\label{lem:pairwise}
    If $H_i(\btheta_0)>H_j(\btheta_0)$, then 
    $$
    P(Y_i<Y_j,\Delta_i=1\mid\bw_i,\bw_j)>P(Y_j<Y_i,\Delta_j=1\mid \bw_i,\bw_j).
    $$
\end{lemma}

\begin{proof}
By definition, 
\begin{align*}
\begin{split}
        &P(Y_i<Y_j,\Delta_i=1\mid \bw_i, \bw_j)\\
=&P\{T_i<T_j,T_i<\min(C_i,C_j)\mid \bw_i, \bw_j\}\\
=&\int_0^\infty P\{T_i<T_j,T_i<u\mid \bw_i, \bw_j,\min(C_i,C_j)=u\}dF_{\min(C_i,C_j)}(u)\\
=&\int_0^\infty P\{T_i<T_j,T_i<u\mid \bw_i, \bw_j\}dF_{\min(C_i,C_j)}(u),
\end{split}
\end{align*}
the last equality by the conditional independence in assumption~\ref{as:cens}.
For a fixed $u$, define the truncated event time $T^{(u)}$ with conditional survival function 
$$
P(T^{(u)}>t\mid \bw)=
\begin{cases}
P(T>t\mid \bw), & t<u,\\
0, & t\geq u.
\end{cases}
$$
That is, $T^{(u)}$ has the same conditional distribution as $T$ on $[0,u)$, but has support bounded by $u$. By construction, for all $t$, $P(T_i^{(u)}>t\mid \bw_i)\leq P(T_j^{(u)}>t\mid \bw_j)$. Therefore, the truncated event time $T_i^{(u)}$ and $T_j^{(u)}$ satisfy the same stochastic ordering as the orginal event times.

Since $T_i^{(u)}$ and $T_j^{(u)}$ remain conditionally independent given $(\bw_i,\bw_j)$, the same argument as in the uncensored case yields $P(T_i^{(u)}<T_j^{(u)}\mid \bw_i,\bw_j)>P(T_j^{(u)}<T_i^{u}\mid \bw_i,\bw_j)$.
Moreover, ${T_i^{(u)}<T_j^{(u)}}={T_i<T_j,T_i<u}$, and ${T_j^{(u)}<T_i^{(u)}}={T_j<T_i,T_j<u}$. Hence, $P(T_i<T_j,T_i<u\mid \bw_i,\bw_j)>P(T_J<T_i,T_j<u\mid \bw_i,\bw_j)$. Integrating both sides with respect to $dF_{\min (C_i,C_j)}(u)$ gives $P(Y_i<Y_j,\Delta_i=1\mid \bw_i,\bw_j)>P(Y_j<Y_i,\Delta_j=1\mid \bw_i,\bw_j)$, which completes the proof.
\hfill$\square$
\end{proof}

Given $H_i(\btheta_0) > H_j(\btheta_0)$, Lemma \ref{lem:pairwise} implies $P(Y_i>Y_j,\Delta_j=1)<P(Y_i<Y_j,\Delta_i=1)$. Thus $\btheta_0$ maximizes
\begin{equation}
\Lambda(\btheta) = \E\Big[ P(Y_i < Y_j, \Delta_i = 1)\cdot
\ind\{H_i(\btheta) > H_j(\btheta)\}
+ P(Y_i > Y_j, \Delta_j = 1)\cdot
\ind\{H_i(\btheta) < H_j(\btheta)\}\Big]
\label{eq:C}
\end{equation}
for each set of $(i,j)$, and therefore maximizes $\E\{C_n(\btheta)\}$, which is the average of $\Lambda(\btheta)$ over the $n(n-1)$ ordered pairs. Here the expectation is taken over $(x_i,\bz_i)$ and $(x_j,\bz_j)$, and $H_i(\btheta)=H(x_i,\bz_i;\btheta)$. 

We now use the above to show that the objective function is uniquely maximized at the truth. The key step is the following identification result, which states that the ordering induced by $H(\cdot\,;\btheta)$ determines
$\btheta$.

\begin{lemma}\label{lem:ident}
    Under assumptions (\ref{as:dens})--(\ref{as:theta}), let $\btheta^{*} \in \Theta$ satisfy
\begin{equation}
\operatorname{sign}\big\{H(x_1,\bm{z}_1;\btheta^{*})
- H(x_2,\bm{z}_2;\btheta^{*})\big\}
= \operatorname{sign}\big\{H(x_1,\bm{z}_1;\btheta_0)
- H(x_2,\bm{z}_2;\btheta_0)\big\}
\label{eq:samesign}
\end{equation}
for almost every pair $(x_1,\bm{z}_1)$, $(x_2,\bm{z}_2)$.
Then $\btheta^{*} = \btheta_0$.
\end{lemma}

\begin{proof}
Condition \eqref{eq:samesign} states that $H(\cdot\,;\btheta^{*})$ and $H(\cdot\,;\btheta_0)$ induce the same ordering on the support of $(x_i,\bm{z}_i)$. Hence there exists a strictly increasing function $\psi$ with
\begin{equation}
H(x,\bm{z};\btheta^{*}) = \psi\big\{H(x,\bm{z};\btheta_0)\big\}
\qquad \text{for almost every } (x,\bm{z}).
\label{eq:psi}
\end{equation}

Write $x_c(\bm{z};\btheta) = (\beta_1+1)^{-1}\{\bm{z}^{\top}(\bm{\alpha}_1 - \bm{\alpha}_2) - \beta_0\}$ for the critical point. By assumption (\ref{as:theta}), $\Theta$ is compact and $\beta_1$ is bounded away from $0$, so $\{x_c(\bm{z};\btheta) : \btheta \in \Theta\}$ is bounded for each $\bm{z}$. By assumption (\ref{as:dens}), $x_i$ has an everywhere positive Lebesgue density conditional on $\bm{z}_i$. Consequently, for almost every $\bz$, each of the sets
$$
\mathcal{I}_{-}(\bm{z}) = \big\{x : x <\min\big(x_c(\bm{z};\btheta_0),x_c(\bm{z};\btheta^{*})\big)\big\},
\mathcal{I}_{+}(\bm{z}) = \big\{x : x >\max\big(x_c(\bm{z};\btheta_0), x_c(\bm{z};\btheta^{*})\big)\big\}
$$
contains an interval of positive conditional probability.

\emph{The first branch.} For $x \in \mathcal{I}_{-}(\bm{z})$, the maxima defining $H(x,\bm{z};\btheta_0)$ and $H(x,\bm{z};\btheta^{*})$ are both attained by their first arguments, so \eqref{eq:psi} becomes
$
\psi\big(-x + \bm{z}^{\top}\bm{\alpha}_1^{(0)}\big) = -x + \bm{z}^{\top}\bm{\alpha}_1^{*}.
$
Writing $s = -x + \bm{z}^{\top}\bm{\alpha}_1^{(0)}$, which ranges over an interval as $x$ does, gives $\psi(s) = s + \bm{z}^{\top}(\bm{\alpha}_1^{*} - \bm{\alpha}_1^{(0)})$. Since $\psi$ does not depend on $\bm{z}$, there is a constant $c$ such that $\bm{z}^{\top}(\bm{\alpha}_1^{*} - \bm{\alpha}_1^{(0)}) = c$ for almost every $\bm{z}$, that is, $c \cdot 1 - \bm{z}^{\top}(\bm{\alpha}_1^{*} - \bm{\alpha}_1^{(0)}) = 0$ almost surely. By assumption (\ref{as:rank}), $E[(1,\bm{z}^{\top})^{\top}(1,\bm{z}^{\top})]$ is nonsingular, so $\bm{\alpha}_1^{*} = \bm{\alpha}_1^{(0)}$ and $c = 0$. Thus $\psi$ is the identity on an interval, and since $\psi$ is strictly increasing and satisfies \eqref{eq:psi} throughout, $\psi$ is the identity.

\emph{The second branch.} For $x \in \mathcal{I}_{+}(\bm{z})$, both maxima are attained by their second arguments, and $\psi = \mathrm{id}$ reduces \eqref{eq:psi} to
$
\beta_0^{*} + \beta_1^{*} x + \bm{z}^{\top}\bm{\alpha}_2^{*}
= \beta_0^{(0)} + \beta_1^{(0)} x + \bm{z}^{\top}\bm{\alpha}_2^{(0)}
$
for all $x$ in an interval and almost every $\bm{z}$. Matching the coefficients of $x$ gives $\beta_1^{*} = \beta_1^{(0)}$, and assumption (\ref{as:rank}) applied to the remaining terms gives $\beta_0^{*} = \beta_0^{(0)}$ and $\bm{\alpha}_2^{*} = \bm{\alpha}_2^{(0)}$. Hence $\btheta^{*} = \btheta_0$. \hfill$\square$
\end{proof}

Now let $\btheta^{*} \in \Theta$ with $\btheta^{*} \neq \btheta_0$, and write $\Delta H(\btheta) = H_i(\btheta) - H_j(\btheta)$. Suppressing the conditioning on $(x_i,\bm{z}_i)$ and $(x_j,\bm{z}_j)$, \eqref{eq:C} gives
\begin{align}
\begin{split}
&\Lambda(\btheta_0) - \Lambda(\btheta^{*})\\
= \E&\Big[\big\{P(Y_i < Y_j, \Delta_i = 1) - P(Y_j < Y_i, \Delta_j = 1)\big\}
\cdot \big[\ind\{\Delta H(\btheta_0) > 0\}
- \ind\{\Delta H(\btheta^{*}) > 0\}\big]\Big].
\label{eq:Cdiff}
\end{split}
\end{align}
By Lemma \ref{lem:pairwise}, the sign of $P(Y_i < Y_j, \Delta_i = 1) - P(Y_j < Y_i, \Delta_j = 1)$ agrees with the sign of $\Delta H(\btheta_0)$, which is nonzero almost surely by assumptions (\ref{as:eps}) and (\ref{as:dens}). Hence the integrand in \eqref{eq:Cdiff} equals
$
\big|P(Y_i < Y_j, \Delta_i = 1) - P(Y_j < Y_i, \Delta_j = 1)\big|
\cdot \ind\big\{\operatorname{sign}\Delta H(\btheta_0)
\neq \operatorname{sign}\Delta H(\btheta^{*})\big\},
$
which is nonnegative, and strictly positive on the event that $\btheta_0$ and $\btheta^{*}$ order a pair of observations differently. By Lemma \ref{lem:ident}, since $\btheta^{*} \neq \btheta_0$, condition \eqref{eq:samesign} must fail, so this event has positive probability. Therefore $\Lambda(\btheta_0) > \Lambda(\btheta^{*})$ for every $\btheta^{*} \neq \btheta_0$, and since $\E\{C_n(\btheta)\}$ is a positive multiple of $\Lambda(\btheta)$, the objective function is uniquely maximized at the truth.

Next, we establish uniform convergence and continuity. The kernel of the U-statistic $C_n(\btheta)$ is bounded by $1$, and the class of kernels indexed by $\btheta \in \Theta$ is shown in Web Appendix \ref{sec:euclidean} to be Euclidean for the constant envelope $1$. Together with the compactness of $\Theta$ in assumption (\ref{as:theta}), Corollary 7 of \cite{sherman1993} gives
$
\sup_{\btheta \in \Theta}
\big|C_n(\btheta) - \Lambda(\btheta)\big| \xrightarrow{P} 0 .
$

We next show that $\Lambda(\btheta)$ is continuous on $\Theta$. Fix $\btheta \in \Theta$ and condition on $\bm{z}_i$, $x_j$ and $\bm{z}_j$. As a function of $x_i$, $H(x_i,\bm{z}_i;\btheta)$ is piecewise linear with slopes $-1$ and $\beta_1$, both of which are bounded away from zero by assumption (\ref{as:theta}); hence it is strictly monotone on each branch, and the equation $H(x_i,\bm{z}_i;\btheta) = H(x_j,\bm{z}_j;\btheta)$ has at most two solutions in $x_i$. Since $x_i$ has an everywhere positive Lebesgue density conditional on $\bm{z}_i$ by assumption (\ref{as:dens}), it follows that $\Pr\{\Delta H(\btheta) = 0\} = 0$ for every $\btheta \in \Theta$. The map $\btheta \mapsto \Delta H(\btheta)$ is continuous, so on the event $\{\Delta H(\btheta) \neq 0\}$ the indicators in \eqref{eq:C} are locally constant in $\btheta$, and the integrand of \eqref{eq:C} is therefore continuous at $\btheta$ almost surely. As it is bounded by $1$, the dominated convergence theorem gives the continuity of $\Lambda(\btheta)$.

Since $\widehat{\btheta}_n$ maximizes $C_n(\btheta)$ over the compact set $\Theta$, $\Lambda(\btheta)$ is continuous and uniquely maximized at $\btheta_0$, and $C_n$ converges uniformly to $\Lambda$ on $\Theta$, standard arguments for extremum estimators (Theorem 2.1 of \cite{newey1994large}) give $\widehat{\btheta}_n \xrightarrow{P} \btheta_0$. This establishes consistency.

\subsection{Establishing the asymptotic normality of a general estimator}\label{sec:general}

In this section, we present two lemmas that will be used to establish the $\sqrt{n}$-consistency and asymptotic normality of the MCE. Both are stated for a general extremum estimator and require no assumptions on the model; the conditions needed to apply them to $C_n(\btheta)$ are verified in Web Appendix \ref{sec:normality}. The lemmas are due to \cite{sherman1993} (Theorems 1 and 2); we reproduce them here, with proofs, for completeness.
Let $\Theta$ be a subset of $\mathbb{R}^m$, and $\btheta_0$ an element of $\Theta$ and a parameter of interest. Suppose $\btheta_0$ maximizes a function $\Gamma(\btheta)$ defined on $\Theta$. Suppose further that a sample analogue, $\Gamma_n(\btheta)$, is maximized at a point $\btheta_n$ that converges in probability to $\btheta_0$.

A brief word on notation. We will be needing uniform bounds on functions of $\btheta$ within shrinking neighborhoods of $\btheta_0$. A convenient notation will be ``$H_n(\btheta)=$ $o_p(1 / n)$ uniformly over $o_p(1)$ neighborhoods of $\btheta_0$." This means that for each sequence of random variables $\left\{r_n\right\}$ of order $o_p(1)$ there exists a sequence of random variables $\left\{b_n\right\}$ of order $o_p(1)$ such that
$$
\sup _{\left|\btheta-\btheta_0\right| \leq r_n}\left|H_n(\btheta)\right| \leq b_n / n
$$
Also, for simplicity, we will assume that $\btheta_0$ is the zero vector (denoted $\mathbf{0}$) in $\mathbb{R}^m$, and that $\Gamma_n(\btheta_0)=\Gamma(\btheta_0)=0$. This entails no loss of generality, since one may always replace $\Gamma_n(\btheta)$ by $\Gamma_n(\btheta_0+\btheta)-\Gamma_n(\btheta_0)$ and $\Gamma(\btheta)$ by $\Gamma(\btheta_0+\btheta)-\Gamma(\btheta_0)$, provided $\btheta_0+\btheta\in\Theta$.

\begin{lemma}\label{lem:rootn}
    Let $\btheta_n$ be a maximizer of $\Gamma_n(\btheta)$, and $\mathbf{0}$ a maximizer of $\Gamma(\btheta)$. Suppose $\btheta_n$ converges in probability to $\mathbf{0}$, and also that
    
(i) there exists a neighborhood $\mathscr{N}$ of $\mathbf{0}$ and a constant $\kappa>0$ for which
$
\Gamma(\btheta) \leq-\kappa|\btheta|^2
$
for all $\btheta$ in $\mathscr{N}$;

(ii) uniformly over $o_p(1)$ neighborhoods of $\mathbf{0}$,
$$
\Gamma_n(\btheta)=\Gamma(\btheta)+O_p(|\btheta| / \sqrt{n})+o_p\left(|\btheta|^2\right)+O_p(1 / n)
$$
Then
$
\left|\btheta_n\right|=O_p(1 / \sqrt{n}).
$
\end{lemma}

\begin{proof}
Since $\btheta_n$ maximizes $\Gamma_n(\btheta)$,
\begin{align}
    0 \leq \Gamma_n\left(\btheta_n\right).\label{eq:maxn}
\end{align}
Since $\btheta_n$ is consistent, it lies within at least one of the sequences of neighborhoods described in (i) and (ii) with probability tending to one as $n$ tends to infinity. When this happens, (i) and (ii) hold for $\btheta=\btheta_n$. Deduce from (\ref{eq:maxn}), (ii), and (i) that
$$
0 \leq-\kappa\left|\btheta_n\right|^2+o_p\left(\left|\btheta_n\right|^2\right)+O_p\left(\left|\btheta_n\right| / \sqrt{n}\right)+O_p(1/n) .
$$
With probability tending to one, the $o_p\left(\left|\btheta_n\right|^2\right)$ term is bounded in absolute value by $\frac{1}{2} \kappa\left|\btheta_n\right|^2$. 
Write $\xi_n\left|\btheta_n\right|$ for the $O_p\left(\left|\btheta_n\right| / \sqrt{n}\right)$ term, where $\left\{\xi_n\right\}$ is a sequence of random variables of order $O_p(1 / \sqrt{n})$. Then, absorbing all things that happen with probability tending to zero into the $O_p(1 / n)$ term, we get
$
2^{-1} \kappa\left|\btheta_n\right|^2-\xi_n\left|\btheta_n\right| \leq O_p(1 / n).
$
Complete the square in $\left|\btheta_n\right|$, to rewrite this as
$
2^{-1} \kappa\left(\left|\btheta_n\right|-\xi_n / \kappa\right)^2 \leq O_p(1 / n)+2^{-1} \xi_n^2 / \kappa=O_p(1 / n) .
$
Take square roots, then rearrange to get
$
\left|\btheta_n\right| \leq \xi_n / \kappa+O_p(1 / \sqrt{n})=O_p(1 / \sqrt{n}).
$
Once $\sqrt{n}$-consistency of $\btheta_n$ is established, we can prove asymptotic normality provided there exist very good quadratic approximations to $\Gamma_n(\btheta)$ within $O_p(1 / \sqrt{n})$ neighborhoods of $\mathbf{0}$. In the following lemma, the symbol $\Rightarrow$ denotes convergence in distribution. \hfill$\square$
\end{proof}

\begin{lemma}\label{lem:normal}
    Suppose $\btheta_n$ is $\sqrt{n}$-consistent for $\mathbf{0}$, an interior point of $\Theta$. Suppose also that uniformly over $O_p(1 / \sqrt{n})$ neighborhoods of $\mathbf{0}$,
\begin{align}
    \quad \Gamma_n(\btheta)=\frac{1}{2} \btheta^{\T} \bV \btheta+\frac{1}{\sqrt{n}} \btheta^{\T} \bW_n+o_p(1/n)\label{eq:quad}
\end{align}
where $\bV$ is a symmetric, negative definite matrix, and $\bW_n$ converges in distribution to a $N(\mathbf{0}, \bU)$ random vector. Then
$$
\sqrt{n} \btheta_n \Rightarrow N\left(\mathbf{0}, \bV^{-1} \bU \bV^{-1}\right) .
$$
\end{lemma}

\begin{proof}
Write $\bt_n$ for $\sqrt{n} \btheta_n$ and $\bt_n^*$ for $-\bV^{-1} \bW_n$. Notice that $\bt_n^* / \sqrt{n}$ maximizes the quadratic approximation to $\Gamma_n(\btheta)$ given in (\ref{eq:quad}). Also, $\bt_n^*$ converges in distribution to a $N\left(\mathbf{0}, \bV^{-1} \bU \bV^{-1}\right)$ random variable. We will show that $\bt_n=\bt_n^*+o_p(1)$.

Because $\mathbf{0}$ is an interior point of $\Theta$, the point $\bt_n^* / \sqrt{n}$ lies in $\Theta$ with probability tending to one. When this happens, by definition of $\bt_n$,
$
\Gamma_n\left(\bt_n^* / \sqrt{n}\right) \leq \Gamma_n\left(\bt_n / \sqrt{n}\right) .
$
Apply (\ref{eq:quad}) twice in the last expression, then multiply through by $n$, consolidate terms, and use the fact that $\bV$ is negative definite to get
$
0 \leq-\frac{1}{2}\left(\bt_n-\bt_n^*\right)^{\T} \bV\left(\bt_n-\bt_n^*\right) \leq o_p(1)
$
The $o_p(1)$ term can be assumed to absorb the bad cases where $\bt_n^* / \sqrt{n}$ does not lie in $\Theta$. The last inequality is true without restriction, and implies that $\bt_n=\bt_n^*+o_p(1)$. \hfill$\square$
\end{proof}

The conditions of the lemmas do not require that $\btheta_n$ be a zero of the gradient of $\Gamma_n(\btheta)$. Nor do they require that $\Gamma_n(\btheta)$ be a continuous function of $\btheta$. Thus, this approach provides a framework within which the asymptotic distribution of the MCE can be established.

\subsection{U-statistic decomposition}\label{sec:ustat}
The next part presents a $U$-statistic decomposition and a uniform bound that are used in tandem with the general method of the last section to establish the asymptotic distribution of the MCE.

Let $\bW_1, \ldots, \bW_n$ be i.i.d. random vectors with distribution $P$ on a set $S$. Let $\Theta$ be a subset of $\mathbb{R}^m$, and for each $\btheta$ in $\Theta$ suppose $f(\cdot, \cdot, \btheta)$ is a real-valued function on the product space $S \otimes S$. Define
$$
U_n f(\cdot, \cdot, \btheta)=\frac{1}{n(n-1)} \sum_{i \neq j} f\left(\bW_i, \bW_j, \btheta\right) .
$$
For each $\btheta$ in $\Theta, U_n f(\cdot, \cdot, \btheta)$ is a $U$-statistic of order two. (See Chapter 5 of \cite{serfling1980} for more on $U$-statistics.) The collection $\left\{U_n f(\cdot, \cdot, \btheta): \btheta \in \Theta\right\}$ is called a $U$-process of order two.

The following decomposition is due to \cite{sherman1993}(Section 3); we restate it in the notation of this paper. By analogy with the empirical measure $P_n$ that places mass $1 / n$ on each $\bW_i, U_n$ can be viewed as a random measure putting mass $1 /[n(n-1)]$ on each ordered pair $\left(\bW_i, \bW_j\right)$. Let $Q$ denote the product measure $P \otimes P$. Then
\begin{align}
    U_n f(\cdot, \cdot, \btheta)=Q f(\cdot, \cdot, \btheta)+P_n g(\cdot, \btheta)+U_n h(\cdot, \cdot, \btheta)\label{eq:decomp}
\end{align}

where, for each $u, v$ in $S$ and each $\theta$ in $\Theta$,
$
g(u, \btheta)=\operatorname{Pf}(u, \cdot, \btheta)+P f(\cdot, u, \btheta)-2 Q f(\cdot, \cdot, \btheta)
$
and
$
h(u, v, \btheta)=f(u, v, \btheta)-P f(u, \cdot, \btheta)-P f(\cdot, v, \btheta)+Q f(\cdot, \cdot, \btheta) .
$
Linear functional notation is used for expectations. $Q f(\cdot, \cdot, \btheta)$ denotes the unconditional expectation of $f(u, v, \btheta)$, while $P f(u, \cdot, \btheta)$ denotes the conditional expectation of $f(u, v, \btheta)$ given its first argument, and $P f(\cdot, v, \btheta)$ the conditional expectation of $f(u, v, \btheta)$ given its second argument.

Notice that for each $\btheta$ in $\Theta, P_n g(\cdot, \btheta)$ is an average of zero-mean, iid random variables. The collection $\left\{P_n g(\cdot, \btheta): \btheta \in \Theta\right\}$ is called a zero-mean empirical process. Also note that
\begin{align}
    Ph(u, \cdot, \btheta) \equiv Ph(\cdot, v, \btheta) \equiv 0.\label{eq:degen}
\end{align}

Because of \eqref{eq:degen}, the function $h(u, v, \btheta)$ is said to be $P$-degenerate on $S \otimes S$ and $U_n h(\cdot, \cdot, \btheta)$ is called a degenerate $U$-statistic of order two. The collection $\{h(\cdot, \cdot, \btheta): \btheta \in \Theta\}$ is said to be a $P$-degenerate class of functions on $S \otimes S$ and $\left\{U_n h(\cdot, \cdot, \btheta): \btheta \in \Theta\right\}$ is called a degenerate $U$-process of order two.

\subsection{Establishing the asymptotic normality of the MCE}\label{sec:normality}
In this section, the MCE is shown to be $\sqrt{n}$-consistent and asymptotically normal. 
Following \citet{sherman1993}, the decomposition in \eqref{eq:decomp} is applied to write $\Gamma_n(\btheta)$ as a sum of its expected value, plus a smoothly parameterized, zero-mean empirical process, plus a degenerate $U$-process of order two. The result is obtained by handling the first two terms using standard Taylor expansion arguments, and then showing that the degenerate term has order $o_p(1 / n)$ uniformly over $o_p$ (1) neighborhoods of the parameter of interest. The following lemma is used to establish the uniformity result.

\begin{lemma}\label{lem:degen}
    Let $\mathscr{F}=\{f(\cdot, \cdot, \btheta): \btheta \in \Theta\}$ be a class of $P$-degenerate functions on $S \otimes S$. Let $Q$ denote the product measure $P \otimes P$. Suppose there exists a point $\btheta_0$ in $\Theta$ for which $f\left(\cdot, \cdot, \btheta_0\right) \equiv 0$. If
    
(i) $\mathscr{F}$ is Euclidean for a constant envelope,

(ii) $Q f(\cdot, \cdot, \btheta)^2 \rightarrow 0$ as $\btheta \rightarrow \btheta_0$, 

then
$
U_n f(\cdot, \cdot, \btheta)=o_p(1 / n)
$
uniformly over $o_p(1)$ neighborhoods of $\btheta_0$.
\end{lemma}

Lemma \ref{lem:degen} is Theorem 3 of \cite{sherman1993}. \citet{pakes1989} provide simple criteria for determining the Euclidean property. \cite{nolan1987} give complementary criteria.

Denote $\btheta_n$ as the maximizer of $C_n(\btheta)$ over $\Theta$. The diagonal terms of $C_n(\btheta)$ in equation (2) of the main text vanish, and the two terms in its summand are exchanges of one another under $i \leftrightarrow j$; hence
\begin{align}
\begin{split}
&C_n(\btheta) = 2\left(1 - n^{-1}\right) U_n f(\cdot,\cdot,\btheta),\\
&f(w_1, w_2, \btheta) =\ind\{y_1 > y_2,\ \delta_2 = 1\}\,
\ind\{H(x_1,\bz_1;\btheta) < H(x_2,\bz_2;\btheta)\},
\label{eq:CnUn}
\end{split}
\end{align}
where $U_n f$ is as in Web Appendix \ref{sec:ustat}. The factor $2(1 - n^{-1})$ does not depend on $\btheta$, so $\btheta_n$ also maximizes $U_n f(\cdot,\cdot,\btheta)$ over $\Theta$; we work with the latter throughout. Let
$$
C(\btheta) = Qf(\cdot,\cdot,\btheta)= \E[\ind\{Y_1 > Y_2,\ \Delta_2 = 1\}\,\ind\{H(X_1,\bZ_1;\btheta) < H(X_2,\bZ_2;\btheta)\}]
$$
denote its expectation, so that $E\{C_n(\btheta)\} = 2(1 - 1/n)\,C(\btheta)$.

For each $\btheta$ in $\Theta$, write $\Gamma(\btheta)$ for $C(\btheta) - C(\btheta_0)$ and $\Gamma_n(\btheta)$ for $U_n f(\cdot,\cdot,\btheta) - U_n f(\cdot,\cdot,\btheta_0)$. Note that $\btheta_n$ maximizes $\Gamma_n(\btheta)$ over $\Theta$, and that $\E\{\Gamma_n(\btheta)\} = \Gamma(\btheta)$. As in Web Appendix \ref{sec:general}, we shall assume that $\btheta_0 = \mathbf{0}$, the zero vector in $\mathbb{R}^{2+2p}$. Thus $\Gamma_n(\mathbf{0}) = \Gamma(\mathbf{0}) = 0$.

In the course of proving the consistency of $\btheta_n$ in Web Appendix \ref{sec:consistency}, we showed that $\Lambda(\btheta)$ is uniquely maximized at $\btheta_0$. Since $\Lambda(\btheta)=2C(\btheta)$, the function $C(\btheta)$ is also uniquely maximized at $\btheta_0$, and it follows immediately that $\Gamma(\btheta)$ is maximized at $\mathbf{0}$.

Recall that $W = (Y, \Delta, X, \bZ)$ denotes an observation from the distribution $P$ on the set $S \subseteq \mathbb{R} \otimes \{0,1\} \otimes \mathbb{R} \otimes \mathbb{R}^p$, and that $W_1, \ldots, W_n$ denotes a sample of independent observations from $P$. For each $w = (y, \delta, x, \bz)$ in $S$ and each $\btheta$ in $\Theta$, define
\begin{align*}
\tau(w, \btheta)
&= \E\Big[\ind\{Y < y,\ \Delta = 1\}\,
   \ind\{H(x,\bz;\btheta) < H(X,\bZ;\btheta)\}\Big] \notag\\
&\quad + \delta \cdot \E\Big[\ind\{Y > y\}\,
   \ind\{H(X,\bZ;\btheta) < H(x,\bz;\btheta)\}\Big],
\end{align*}
which is the sum $Pf(w,\cdot,\btheta) + Pf(\cdot,w,\btheta)$ of the two conditional expectations of the kernel $f$ in (\ref{eq:CnUn}). The function $\tau(\cdot,\btheta)$ will be the kernel of the empirical process that drives the asymptotic behavior of $\btheta_n$. Notice that $E[\tau(\cdot,\btheta) - \tau(\cdot,\mathbf{0})] = 2\Gamma(\btheta)$. Write $\nabla_m$ for the $m$th partial derivative operator with respect to $\btheta$, and
$$
|\nabla_m|\sigma(\btheta) \equiv \sum_{i_1,\ldots,i_m}
\left|\frac{\partial^m}{\partial \theta_{i_1} \cdots \partial \theta_{i_m}}
\sigma(\btheta)\right| .
$$
The symbol $\|\cdot\|$ denotes the matrix norm $\|(a_{ij})\| = (\sum_{i,j} a_{ij}^2)^{1/2}$.

We now state the additional assumptions used in the normality proof. Let
$\mathscr{N}$ denote a neighborhood of $\mathbf{0}$.

\noindent\textbf{Assumption (7).}\hspace{0.5em}
\refstepcounter{as}\label{as:reg}%
In addition to assumptions (\ref{as:G})--(\ref{as:theta}) of Web Appendix \ref{sec:assump}, suppose that $\mathbf{0}$ is an interior point of $\Theta$ and that
    \begin{itemize}[label={},leftmargin=2em,itemsep=2pt]
        \item[(i)]\label{as:reg-i} For each $w$ in $S$, all mixed second partial derivatives of $\tau(w,\cdot)$ exist on $\mathscr{N}$.
        \item[(ii)]\label{as:reg-ii} There is an integrable function $M(w)$ such that for all $w$ in $S$ and $\btheta$ in $\mathscr{N}$,
$$
\|\nabla_2 \tau(w, \btheta) - \nabla_2 \tau(w, \mathbf{0})\| \leq M(w)\,|\btheta| .
$$
        \item[(iii)]\label{as:reg-iii} $\E|\nabla_1 \tau(\cdot, \mathbf{0})|^2 < \infty$.
        \item[(iv)]\label{as:reg-iv} $\E|\nabla_2|\tau(\cdot, \mathbf{0}) < \infty$.
        \item[(v)]\label{as:reg-v} The matrix $\E\nabla_2 \tau(\cdot, \mathbf{0})$ is negative definite.
\end{itemize}

The conditions of assumption (\ref{as:reg}) are standard regularity conditions sufficient to support an argument based on a Taylor expansion of $\tau(w, \cdot)$ about $\mathbf{0}$; part (v) is what makes the limiting variance well defined. Sufficient conditions on the distribution of $W$ for assumption (\ref{as:reg}) to hold are discussed in Web Appendix \ref{sec:UV}.

\begin{theorem}
If assumptions (\ref{as:G})--(\ref{as:reg}) hold, then
$
\sqrt{n}\,\btheta_n \Rightarrow N\left(\mathbf{0},\ \bV^{-1}\bU\bV^{-1}\right),
$
where $2\bV = E\nabla_2\tau(\cdot,\mathbf{0})$ and
$\bU = E\big[\nabla_1\tau(\cdot,\mathbf{0})\nabla_1\tau(\cdot,\mathbf{0})^{\T}\big]$.
\end{theorem}

\begin{proof}
We will show that
\begin{align}
\Gamma_n(\btheta) = \frac{1}{2}\btheta^{\T}\bV\btheta
+ \frac{1}{\sqrt{n}}\btheta^{\T}\bW_n
+ o_p\left(|\btheta|^2\right) + o_p(1/n)
\label{eq:master}
\end{align}
uniformly in $o_p(1)$ neighborhoods of $\mathbf{0}$, where $\bW_n$ converges in distribution to a $N(\mathbf{0}, \bU)$ random vector.

We first check the conditions of Lemma \ref{lem:rootn}. By (\ref{eq:Gamma}) below and the negative definiteness of $\bV$ in assumption (\ref{as:reg-v}), there exist a neighborhood $\mathscr{N}$ of $\mathbf{0}$ and a constant $\kappa > 0$ such that $\Gamma(\btheta) \leq -\kappa|\btheta|^2$ for all $\btheta$ in $\mathscr{N}$, which is condition (i). Subtracting (\ref{eq:Gamma}) from (\ref{eq:master}) gives
$$
\Gamma_n(\btheta) = \Gamma(\btheta) + O_p(|\btheta|/\sqrt{n})
+ o_p\left(|\btheta|^2\right) + o_p(1/n)
$$
uniformly over $o_p(1)$ neighborhoods of $\mathbf{0}$, which is condition (ii). Lemma \ref{lem:rootn} therefore gives
\begin{align}
\left|\btheta_n\right| = O_p(1/\sqrt{n}).
\label{eq:rootn}
\end{align}
By (\ref{eq:rootn}), $|\btheta|^2 = O_p(1/n)$ uniformly over $O_p(1/\sqrt{n})$ neighborhoods of $\mathbf{0}$, so the term $o_p(|\btheta|^2)$ in (\ref{eq:master}) is $o_p(1/n)$ there. Hence (\ref{eq:master}) reduces to the expansion required by Lemma \ref{lem:normal}, and the result follows from Lemma \ref{lem:normal}.

It remains to establish (\ref{eq:master}).

For each $(w_1, w_2)$ in $S \otimes S$ and each $\btheta$ in $\Theta$, define the centered kernel
\begin{align*}
\tilde{f}&(w_1, w_2, \btheta) = f(w_1, w_2, \btheta) - f(w_1, w_2, \mathbf{0})\\
= &\ind\{y_1 > y_2,\ \delta_2 = 1\}\Big[
\ind\{H(x_1,\bz_1;\btheta) < H(x_2,\bz_2;\btheta)\}
- \ind\{H(x_1,\bz_1;\mathbf{0}) < H(x_2,\bz_2;\mathbf{0})\}\Big],
\end{align*}
so that $\Gamma_n(\btheta) = U_n\tilde{f}(\cdot,\cdot,\btheta)$, $Q\tilde{f}(\cdot,\cdot,\btheta) = \Gamma(\btheta)$, and $\tilde{f}(\cdot,\cdot,\mathbf{0}) \equiv 0$. Since $\Gamma_n(\btheta)$ is a $U$-statistic of order two with expectation $\Gamma(\btheta)$, we may apply the decomposition in \eqref{eq:decomp} to write
$$
\Gamma_n(\btheta) = \Gamma(\btheta) + P_n g(\cdot,\btheta)
+ U_n h(\cdot,\cdot,\btheta),
$$
where
$g(w, \btheta) = P\tilde{f}(w,\cdot,\btheta) + P\tilde{f}(\cdot,w,\btheta)
- 2\Gamma(\btheta)$, and
$$
h(w_1, w_2, \btheta) = \tilde{f}(w_1,w_2,\btheta)
- P\tilde{f}(w_1,\cdot,\btheta) - P\tilde{f}(\cdot,w_2,\btheta)
+ \Gamma(\btheta).
$$

First, we show that
\begin{align}
\Gamma(\btheta) = \frac{1}{2}\btheta^{\T}\bV\btheta
+ o\left(|\btheta|^2\right) \qquad \text{as } \btheta \to \mathbf{0}.
\label{eq:Gamma}
\end{align}
Fix $w$ in $S$ and $\btheta$ in $\mathscr{N}$. Invoke assumption (\ref{as:reg-i}) and expand $\tau(w,\btheta)$ about $\mathbf{0}$ to get
\begin{align}
\tau(w,\btheta) = \tau(w,\mathbf{0}) + \btheta^{\T}\nabla_1\tau(w,\mathbf{0})
+ \frac{1}{2}\btheta^{\T}\nabla_2\tau(w,\btheta^*)\btheta
\label{eq:taylor}
\end{align}
for $\btheta^*$ between $\btheta$ and $\mathbf{0}$. By assumption (\ref{as:reg-ii}), for each $w$ in $S$ and each $\btheta$ in $\mathscr{N}$,
\begin{align}
\left\|\btheta^{\T}\left[\nabla_2\tau(w,\btheta)
- \nabla_2\tau(w,\mathbf{0})\right]\btheta\right\| \leq M(w)|\btheta|^3 .
\label{eq:lip}
\end{align}
Take expectations in (\ref{eq:taylor}), and apply (\ref{eq:lip}), the integrability of $M$, and
$E[\tau(\cdot,\btheta) - \tau(\cdot,\mathbf{0})] = 2\Gamma(\btheta)$, to get that
$
2\Gamma(\btheta) = \btheta^{\T}E\nabla_1\tau(\cdot,\mathbf{0})
+ \btheta^{\T}\bV\btheta + o\left(|\btheta|^2\right),\text{as } \btheta \to \mathbf{0} .
$
Since $\Gamma(\btheta)$ is maximized at $\mathbf{0}$, which is an interior point of $\Theta$ by assumption (\ref{as:reg}), the coefficient of the linear term in the last expression must be the zero vector. Divide through by 2 to establish (\ref{eq:Gamma}).

Next, we show that
\begin{align}
P_n g(\cdot,\btheta) = \frac{1}{\sqrt{n}}\btheta^{\T}\bW_n
+ o_p\left(|\btheta|^2\right)
\label{eq:Png}
\end{align}
uniformly over $o_p(1)$ neighborhoods of $\mathbf{0}$, where $\bW_n$ converges in distribution to a $N(\mathbf{0}, \bU)$ random vector with $\bU = E[\nabla_1\tau(\cdot,\mathbf{0})\nabla_1\tau(\cdot,\mathbf{0})^{\T}]$. Note that
$$
g(w,\btheta) = \tau(w,\btheta) - \tau(w,\mathbf{0}) - 2\Gamma(\btheta).
$$
Apply (\ref{eq:Gamma}), (\ref{eq:taylor}), and (\ref{eq:lip}) to see that
\begin{align}
P_n g(\cdot,\btheta) = \frac{1}{\sqrt{n}}\btheta^{\T}\bW_n
+ \frac{1}{2}\btheta^{\T}\bD_n\btheta + o\left(|\btheta|^2\right)
+ R_n(\btheta)
\label{eq:Png2}
\end{align}
uniformly over $o_p(1)$ neighborhoods of $\mathbf{0}$, where $\bW_n = \sqrt{n}\,P_n\nabla_1\tau(\cdot,\mathbf{0})$, $\bD_n = P_n\nabla_2\tau(\cdot,\mathbf{0}) - 2\bV$,  $\left|R_n(\btheta)\right| \leq |\btheta|^3\, P_n M(\cdot)$.
Deduce from assumption (\ref{as:reg-iii}) and the fact, established above, that $E\nabla_1\tau(\cdot,\mathbf{0}) = \mathbf{0}$, that $\bW_n$ converges in distribution to a $N(\mathbf{0}, \bU)$ random vector. By assumption (\ref{as:reg-iv}) and a weak law of large numbers, $\bD_n$ converges in probability to the zero matrix as $n$ tends to infinity, so that $\frac{1}{2}\btheta^{\T}\bD_n\btheta = o_p(|\btheta|^2)$. Finally, deduce from the integrability of $M$ and a weak law of large numbers that $R_n(\btheta) = o_p(|\btheta|^2)$ uniformly over $o_p(1)$ neighborhoods of $\mathbf{0}$. This establishes (\ref{eq:Png}).

In order to establish (\ref{eq:master}), it remains to show that
\begin{align}
U_n h(\cdot,\cdot,\btheta) = o_p(1/n)
\label{eq:Unh}
\end{align}
uniformly over $o_p(1)$ neighborhoods of $\mathbf{0}$. Lemma \ref{lem:degen} will do the job.

Web Appendix~\ref{sec:euclidean} shows that the class $\mathscr{F} = \{f(\cdot,\cdot,\btheta) : \btheta \in \Theta\}$ is Euclidean for the constant envelope 1. The Euclidean property of the $P$-degenerate class $\{h(\cdot,\cdot,\btheta) : \btheta \in \Theta\}$ follows from this fact in combination with Corollaries 17 and 21 of \citet{nolan1987}, which is condition (i) of Lemma~\ref{lem:degen}. Equation~\eqref{eq:Unh} will therefore follow from Lemma~\ref{lem:degen} provided
\begin{align}
Qh(\cdot,\cdot,\btheta)^2 \to 0 \qquad \text{as } \btheta \to \mathbf{0},
\label{eq:Qh}
\end{align}
where $Q$ is the product measure $P \otimes P$.

By assumption (\ref{as:theta}), $H(\cdot,\bz;\btheta)$ is piecewise linear in $x$ with slopes $-1$ and $\beta_1$, both bounded away from zero, and hence strictly monotone on each of its two branches. By assumption (\ref{as:dens}), $x$ has an everywhere positive Lebesgue density conditional on $\bz$. It follows that the distribution of $H(X,\bZ;\mathbf{0})$ is absolutely continuous with respect to Lebesgue measure on $\mathbb{R}$, and hence that
$
Q\big\{H(x_1,\bz_1;\mathbf{0}) = H(x_2,\bz_2;\mathbf{0})\big\} = 0 .
$
Deduce that $f(w_1,w_2,\cdot)$ is continuous at $\mathbf{0}$ for $Q$ almost all $(w_1,w_2)$. The boundedness of $f$ and a dominated convergence argument imply that the same holds true for $h(w_1,w_2,\cdot)$. Since $h$ is bounded, another dominated convergence argument establishes (\ref{eq:Qh}), which in turn establishes (\ref{eq:Unh}).

Put it all together. Combine (\ref{eq:Gamma}), (\ref{eq:Png}), and (\ref{eq:Unh}) to get (\ref{eq:master}). This proves Theorem 1. \hfill$\square$
\end{proof}

\subsection{Proof of the Euclidean Property}\label{sec:euclidean}

Consider the class of functions $\mathscr{F}=\{f(\cdot,\cdot,\btheta):\btheta \in \Theta\}$, where, for each $(w_1,w_2)$ in $S \otimes S$ and each $\btheta$ in $\Theta$,
$
    f(w_1,w_2,\btheta)=\ind\{y_1>y_2,\ \delta_2=1\}\,
    \ind\{H(x_1,\bz_1;\btheta)<H(x_2,\bz_2;\btheta)\}.
$

In this part, we show that $\mathscr{F}$ is Euclidean for the constant envelope 1. 

Let $\mathscr{G}$ denote the vector space of all affine functions of $(t,\,y_1,\,y_2,\,\delta_2,\,x_1,\,x_2,\,\bz_1,\,\bz_2)$ on $S\otimes S\otimes\mathbb{R}$; that is, the collection of all functions of the form
$
g(w_1,w_2,t)=\gamma t+\gamma_1 y_1+\gamma_2 y_2+\psi\delta_2
+\lambda_1 x_1+\lambda_2 x_2
+\bkappa_1^{\T}\bz_1+\bkappa_2^{\T}\bz_2+c ,
$
with $\gamma,\gamma_1,\gamma_2,\psi,\lambda_1,\lambda_2,c\in\mathbb{R}$ and $\bkappa_1,\bkappa_2\in\mathbb{R}^p$. Notice that $\mathscr{G}$ is a $(2p+7)$-dimensional vector space of real-valued functions on $S\otimes S\otimes\mathbb{R}$. By Lemma 2.4 in \citet{pakes1989}, the class of sets of the form $\{g\geq r\}$ or $\{g>r\}$, with $g\in\mathscr{G}$ and $r\in\mathbb{R}$, is a VC class. We use this fact to show that the set of subgraphs of functions belonging to $\mathscr{F}$ forms a VC class of sets.

Since $f$ takes values in $\{0,1\}$ and $\delta_2$ takes values in $\{0,1\}$, for each $\btheta$ in $\Theta$,
\begin{align}
    \begin{split}
        &\text{subgraph}(f(\cdot,\cdot,\btheta))\\
        &=\{(w_1,w_2,t)\in S\otimes S\otimes\mathbb{R}:\ 0<t<f(w_1,w_2,\btheta)\}\\
        &=\{y_1-y_2>0\}\cap\{\delta_2\geq 1\}
        \cap\{H(x_1,\bz_1;\btheta)<H(x_2,\bz_2;\btheta)\}
        \cap\{t>0\}\cap\{1-t>0\}.
    \end{split}\label{eq:subgraph}
\end{align}
Write $a_k=-x_k+\bz_k^{\T}\balpha_1$ and $b_k=\beta_0+\beta_1 x_k+\bz_k^{\T}\balpha_2$ for $k=1,2$, so that $H(x_k,\bz_k;\btheta)=\max(a_k,b_k)$. Because $\max(a_2,b_2)>M$ if and only if $a_2>M$ or $b_2>M$,
$
        \{H(x_1,\bz_1;\btheta)<H(x_2,\bz_2;\btheta)\}
        =\big[\{a_2>a_1\}\cap\{a_2>b_1\}\big]
        \cup\big[\{b_2>a_1\}\cap\{b_2>b_1\}\big],
$
where
\begin{align*}
\{a_2>a_1\}&=\{(x_1-x_2)+(\bz_2-\bz_1)^{\T}\balpha_1>0\},\\
\{a_2>b_1\}&=\{-x_2-\beta_1 x_1-\beta_0
+\bz_2^{\T}\balpha_1-\bz_1^{\T}\balpha_2>0\},\\
\{b_2>a_1\}&=\{x_1+\beta_1 x_2+\beta_0
+\bz_2^{\T}\balpha_2-\bz_1^{\T}\balpha_1>0\},\\
\{b_2>b_1\}&=\{\beta_1(x_2-x_1)+(\bz_2-\bz_1)^{\T}\balpha_2>0\}.
\end{align*}
Each of the four displayed functions is an affine function of $(x_1,x_2,\bz_1,\bz_2)$ and therefore belongs to $\mathscr{G}$; so do $g_1=y_1-y_2$, $g_2=\delta_2$, $g_3=t$, and $g_4=1-t$. By Lemma 2.4 of \citet{pakes1989}, each of the corresponding sets lies in a VC class. The subgraph in \eqref{eq:subgraph} is a finite union of finite intersections of these sets, so by Lemma 2.5 of \citet{pakes1989} the class $\{\text{subgraph}(f(\cdot,\cdot,\btheta)):\btheta\in\Theta\}$ is a VC class of sets. By Lemma 2.12 of \citet{pakes1989}, $\mathscr{F}$ is Euclidean for every
envelope, and in particular for the constant envelope 1, since $|f|\leq 1$.

\subsection{One-Dimensional Forms of the Variance Components $U$ and $V$}\label{sec:UV}

The variance components $U$ and $V$ from Theorem \ref{thm-1} of the main text admit tractable one-dimensional forms by conditioning on the value of the index $H(X_i,\bZ_i;\btheta_0)$ rather than on the full covariate vector; the argument is that of Theorem 4 of \citet{sherman1993}, adapted to the piecewise linear index and to right censoring.

Because $H(\cdot,\bz;\btheta)$ is piecewise linear in $x$, its gradient with respect to $\btheta$ is piecewise constant, with a jump at the critical point:
\begin{align*}
\bPsi(x,\bz) \;=\; \nabla_{\btheta} H(x,\bz;\btheta_0)
\;=\;
\begin{cases}
\big(0,\ 0,\ \bz^{\T},\ \mathbf{0}^{\T}\big)^{\T},
& x < x_c(\bz;\btheta_0),\\[4pt]
\big(1,\ x,\ \mathbf{0}^{\T},\ \bz^{\T}\big)^{\T},
& x > x_c(\bz;\btheta_0).
\end{cases}
\label{eq:Psi}
\end{align*}
By assumption (\ref{as:dens}), $\Pr\{X_i = x_c(\bZ_i;\btheta_0)\} = 0$, so $\bPsi(X_i,\bZ_i)$ is defined almost surely. The $(2+2p)$-vector $\bPsi$ plays the role of the regressor vector in \citet{sherman1993}; it is here that the non-differentiability of $H$ at the critical point enters. By assumptions (\ref{as:dens}) and (\ref{as:theta}), the distribution of $H(X_i,\bZ_i;\btheta_0)$ is absolutely continuous; write $g_0(\cdot)$ for its density.

For a fixed observation $w=(y,\delta,x,\bz)$ and $t \in \mathbb{R}$, define
\begin{align*}
S(y,\delta,t) =
\E\Big[\ind\{y < Y\} - \delta\cdot\ind\{Y < y, \Delta = 1\}
\;\Big|\; H(X,\bZ;\btheta_0) = t\Big],
\end{align*}
where $(Y,\Delta,X,\bZ)$ denotes an independent observation from $P$, and let
\begin{align*}
\mu_0(t) = \E\big[\bPsi(X,\bZ) \mid H(X,\bZ;\btheta_0) = t\big] .
\end{align*}
Since two observations with the same value of the index are exchangeable, $S(\cdot,\cdot,\cdot)$ satisfies the moment condition
$
\E\big[S(Y_i,\Delta_i,t) \mid H(X_i,\bZ_i;\btheta_0) = t\big] = 0,\text{for every } t \in \mathbb{R} .
$
Write $S_3(y,\delta,t)$ for $\partial S(y,\delta,t)/\partial t$.

When $S(\cdot,\cdot,\cdot)$ is differentiable with respect to its third argument, $g_0(\cdot)$ is differentiable, and $\E\|\bPsi(X_i,\bZ_i)\|^2 < \infty$, the same change-of-variables argument as in Theorem 4 of \citet{sherman1993} gives
$
\nabla_1 \tau(w,\btheta_0)
= \big\{\bPsi(x,\bz) - \mu_0\big(H(x,\bz;\btheta_0)\big)\big\}\,
S\big(y,\delta,H(x,\bz;\btheta_0)\big)\,
g_0\big(H(x,\bz;\btheta_0)\big),
$
so that $U$ and $V$ collapse to one-dimensional integrals in the index $H$ with respect to $g_0$. Writing $H_i = H(X_i,\bZ_i;\btheta_0)$ and $\bPsi_i = \bPsi(X_i,\bZ_i)$ for brevity,
\begin{align*}
U &= \E\Big[
\big\{\bPsi_i - \mu_0(H_i)\big\}
\big\{\bPsi_i - \mu_0(H_i)\big\}^{\T}
\cdot S(Y_i,\Delta_i,H_i)^2\,
g_0(H_i)^2 \Big],
\\[4pt]
V &= 2^{-1}\,\E\Big[
\big\{\bPsi_i - \mu_0(H_i)\big\}
\big\{\bPsi_i - \mu_0(H_i)\big\}^{\T}
\cdot S_3(Y_i,\Delta_i,H_i)\,
g_0(H_i) \Big].
\end{align*}

These expressions also show that assumption~(\ref{as:reg})(v) need not be imposed separately. Since $H$ is the risk index, a larger value of $t$ corresponds to a stochastically smaller failure time, so $S(y,\delta,\cdot)$ is decreasing in its third argument and $S_3 < 0$. Writing
$$
\bW_i = \big\{\bPsi_i - \mu_0(H_i)\big\}
\sqrt{-S_3(Y_i,\Delta_i,H_i)\,g_0(H_i)} ,
$$
we obtain $-2V = \E\big[\bW_i \bW_i^{\T}\big]$, which is positive semi-definite. It is positive definite, and hence $V$ is negative definite, provided $\bPsi_i - \mu_0(H_i)$ is not almost surely contained in a proper linear subspace of $\mathbb{R}^{2+2p}$; this holds under assumptions~(\ref{as:dens}) and~(\ref{as:rank}).

These forms directly motivate a plug-in estimator by smoothing $g_0$ and $\mu_0$ in $t$ and estimating $S$ and $S_3$ via local methods; see also \citet{khan2007partial}.

\section*{Web Appendix B: Method to Determine Initial Values for MCE Optimization}\label{sec:web-appendix-B-init}

To complement technical details in applying the optimization step, we develop a method to obtain initial values by leveraging the U-shaped risk patterns, which reduces the computational complexity.

Recall the U-shaped risk component $H(X, \bZ) = \max(-X+\bZ^{\T}\balpha_1, \beta_0 + \beta_1 X + \bZ^{\T}\balpha_2)$ in model~(\ref{semi_model}) of the main text. Under the constraint $\beta_1>0$, $H(X,\bZ)$ is U-shaped in $X$, so for any risk level $a$ strictly between the minimum and the shoulder levels, there exist exactly two biomarker values satisfying $P(T<t|X,\bZ)=a$---one on each branch of the U-shaped curve---which we denote $X_l$ (left of the critical point) and $X_r$ (right of the critical point). Since $P(T<t|X,\bZ)$ is monotone in $H$, these are the solutions to $\max(-X+\bZ^{\T}\balpha_1, \beta_0 + \beta_1 X + \bZ^{\T}\balpha_2)=C_{t,a}$, where $C_{t,a} = G^{-1}(t)-F_\epsilon^{-1}(a)$ is a constant depending on $t$ and $a$, and $F_\epsilon(\cdot)$ is the cumulative distribution function of the error term $\epsilon$. At the minimum, $X_l=X_r$.

We exploit the linear relationship between $X_l$ and $X_r$ to obtain initial values, first in a simplified setting, then in the general case with continuous covariates. In the simplified setting when $\bZ$ is a fixed vector and $\balpha_1$ and $\balpha_2$ are known. $X_l$ and $X_r$ lie on opposite branches of $H$, satisfying $-X_l+\bZ^{\T}\balpha_1=C_{t,a}$ and $\beta_0+\beta_1 X_r+\bZ^{\T}\balpha_2=C_{t,a}$. Eliminating $C_{t,a}$ yields a linear relationship between $X_r$ and $X_l$,
$$X_r=-\beta_1^{-1}\beta_0-\beta_1^{-1}X_l+\beta_1^{-1}(\bZ^{\T}\balpha_1-\bZ^{\T}\balpha_2).$$
Regressing $X_r$ on $X_l$, the slope and intercept estimates supply the ad-hoc approximated values of $\beta_0$ and $\beta_1$, which serve as initial values for our optimizer.

In real-life datasets, $\bZ \in \R^p$ contains continuous covariates. In this case, we extend this approach by discretizing each continuous component at clinically meaningful cut points and treating every unique combination of the discretized covariates as a subgroup, $z=1,\ldots,2^p$. Denoting the paired roots for subgroup $z$ by $X_l(z)$ and $X_r(z)$, the same elimination of $C_{t,a}$ in equations $-X_l(z)+\bZ_z^{\T}\balpha_1=C_{t,a}$ and $\beta_0+\beta_1 X_r(z)+\bZ_z^{\T}\balpha_2=C_{t,a}$ yields a linear relationship between $X_r$, $X_l$ and $\bZ_z$ within each subgroup,
$$X_r(z)=-\beta_1^{-1}\beta_0-\beta_1^{-1}X_l(z)+\beta_1^{-1}(\bZ_z^{\T}\balpha_1-\bZ_z^{\T}\balpha_2).$$
Regressing $X_r(z)$ on $X_l(z)$ and the subgroup indicators $z$ yields estimates of $\beta_0$, $\beta_1$, $\balpha_1$ and $\balpha_2$ as initial values for optimization. To ensure identifiability, we parameterize the model using three free parameters $(\balpha_{2z},\beta_{0z},\beta_{1z})$ for each subgroup $z$, while fixing $\balpha_{1z}=0$. This constraint resolves the non-identifiability arising from additive shifts in $\balpha_1$ and $\beta_0$, ensuring a unique representation of the subgroup-specific risk function.

To estimate $X_l(z)$ and $X_r(z)$ from data, we approximate the risk curve within each subgroup via LOESS-smoothed incidence rates on a binned BMI grid. Within each subgroup, we partition the BMI range into half-open bins $\mathcal{B}=\{[b_k,b_{k+1})\}_{k=1}^K$, where $b_1$ and $b_{K+1}$ are the smallest and largest BMI values, respectively, with each bin having a width of 2.5 units. The midpoint of bin $k$ is defined as $m_k=(b_k+b_{k+1})/2$. For individuals $\mathcal{I}_{zk}$ in bin $k$ and subgroup $z$, with follow-up time $T_i$ and event indicator $\Delta_i$, the incidence rate per 1,000 person-years is $IR_{zk}=1000 \times \frac{E_{zk}}{PY_{zk}}$, where $PY_{zk}=\sum_{i\in\mathcal{I}_{zk}}T_i$, $E_{zk}=\sum_{i\in \mathcal{I}_{zk}}\Delta_i$, computed only for bins with $PY_{zk}>0$. Smoothing the $(m_k, IR_{zk})$ sequence with LOESS regression yields a stable U-shaped empirical curve, from which we extract the paired $(X_l,X_r)$ for subsequent initialization. To improve robustness, we apply random perturbations (jittering) around each paired value and select the best-performing initializations. Specifically, for the $j$-th jitter run, we generate $x^{(j)}_l(h)=x_l(h)+\epsilon_l^{(j)}$, $x_r^{(j)}=x_r(h)+\epsilon_r^{(j)}$, where $\epsilon^{(j)}$ iid follows $N(0,\sigma^2_{BMI})$, with a small $\sigma_{BMI}$ (0.15--0.20 BMI units). Each jittered set yields an initial parameter vector $\theta_{init}^{(j)}$ passed to the COBYLA optimizer. Among all valid runs, we select the one with the largest C-index: $j^*=\text{argmax}_j\text{C-Index}(\hat H_{\theta^{(j)}})$.

\section*{Web Appendix C: Additional Simulation Results}
\label{sec:web-appendix-C}

This appendix provides per-method and full per-scenario numerical results complementing Table~\ref{tab:sim-main} and Figures~\ref{fig:sim-xc}--\ref{fig:sim-cindex} of the main text. Table~\ref{tab:sim-cindex} reports critical-point $X_c$ estimation accuracy and global test-set $C$-index side-by-side for the MCE and the five Cox alternatives across the nine $30\%$-censoring scenarios that anchor the main grid (panel (a) of Figures~\ref{fig:sim-xc}--\ref{fig:sim-cindex}); this method-by-method view is the source for the numerical claims in the Simulation Study section. Table~\ref{tab:supp-cindex-all} reports the test-set Harrell's $C$-index across all $23$ scenarios for the five method families plus the oracle, providing the full per-scenario complement to the visual summary of Figure~\ref{fig:sim-cindex}.

\clearpage
\begin{figure}[!htbp]
    \centering
    \includegraphics[width=\textwidth]{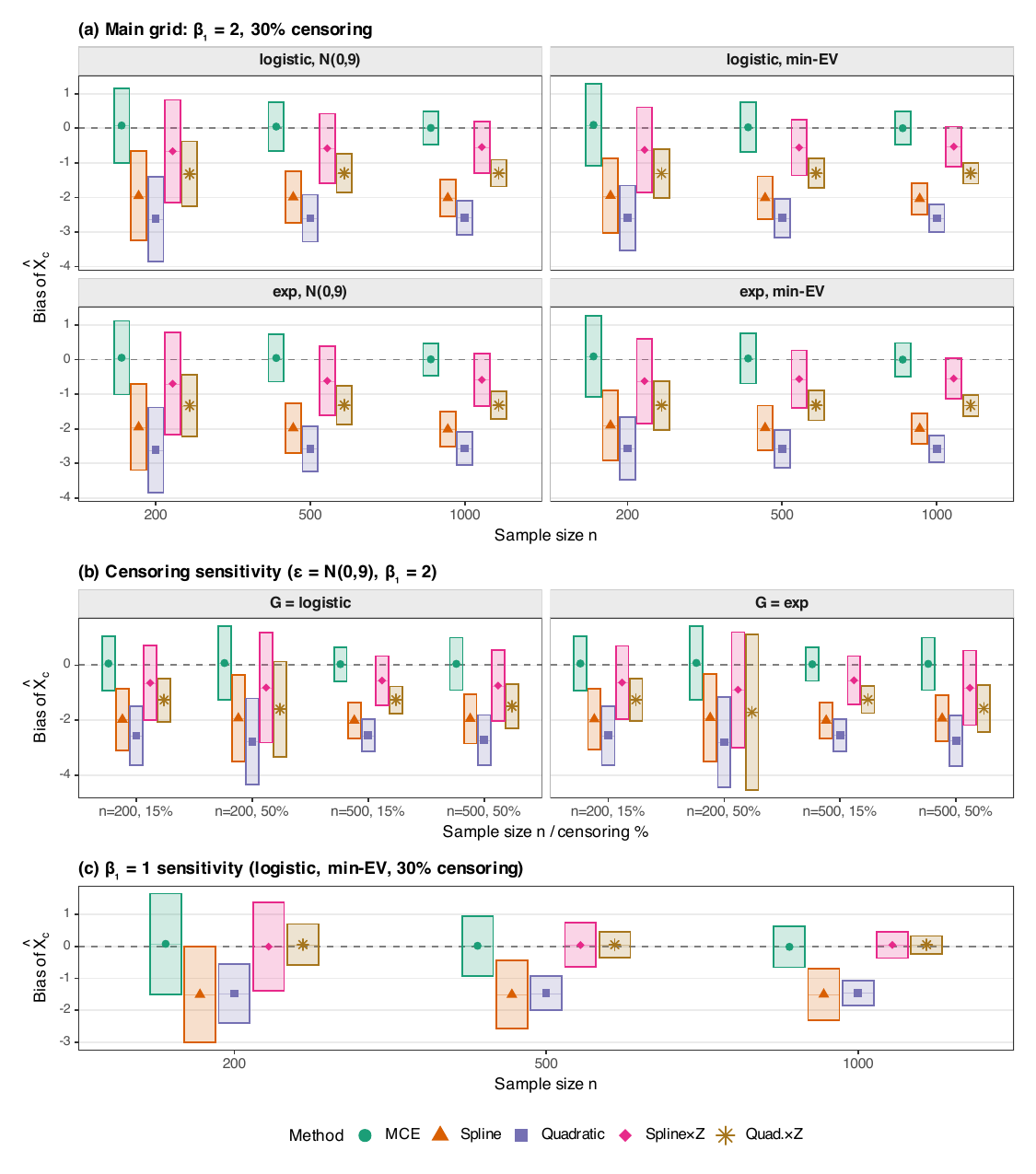}
    \caption{\small\linespread{1.0}\selectfont Bias of $\hat{X}_c$ across all $23$ simulation scenarios and five feasible methods (MCE and four Cox alternatives; the oracle is omitted because it does not estimate $X_c$). Each marker is the across-replicate mean bias and the surrounding box spans $\pm 1.96\times$ empirical standard error (ESE); the horizontal dashed line at $0$ denotes no bias. (a) Main grid: scenarios with $\beta_1 = 2$ and $30\%$ censoring across the four $(G,\epsilon)$ cells, with sample size $n$ on the $x$-axis. (b) Censoring sensitivity at $\epsilon = N(0,9)$ and $\beta_1 = 2$: the additional $15\%$ and $50\%$ censoring scenarios for each $G$ (the matching $30\%$ scenarios appear in panel (a)). (c) Milder U-shape sensitivity at $\beta_1 = 1$, logistic $G$, min-EV error, and $30\%$ censoring. Method colors and marker shapes are shared across panels and across Figure~\ref{fig:sim-cindex}; results are based on $1{,}000$ replications per scenario.}
    \label{fig:sim-xc}
\end{figure}

\begin{sidewaystable}[hb]
\centering
\caption{Critical point $X_c$ estimation and test-set $C$-index across nine scenarios and six methods.}\label{tab:sim-cindex}
\vspace{-6pt}
\setlength{\tabcolsep}{2.5pt}
\renewcommand{\arraystretch}{1.0}
\resizebox{\textwidth}{!}{%
\begin{tabular}{ccc @{\hskip 6pt} ccccc @{\hskip 6pt} ccccc @{\hskip 6pt} ccccc @{\hskip 6pt} ccccc @{\hskip 6pt} cccccc}
\hline
& & & \multicolumn{5}{c}{Bias} & \multicolumn{5}{c}{ESE} & \multicolumn{5}{c}{RMSE} & \multicolumn{5}{c}{CP} & \multicolumn{6}{c}{$C$-index} \\
\cline{4-8}\cline{9-13}\cline{14-18}\cline{19-23}\cline{24-29}
$n$ & $G$ & $\epsilon$ & MCE & Sp & Q & Sp$\!\times\!\!Z$ & Q$\!\times\!\!Z$ & MCE & Sp & Q & Sp$\!\times\!\!Z$ & Q$\!\times\!\!Z$ & MCE & Sp & Q & Sp$\!\times\!\!Z$ & Q$\!\times\!\!Z$ & MCE & Sp & Q & Sp$\!\times\!\!Z$ & Q$\!\times\!\!Z$ & MCE & Sp & Q & Sp$\!\times\!\!Z$ & Q$\!\times\!\!Z$ & Or \\
\hline
$200$ & logistic & $N(0,9)$ & $\mathbf{+0.08}$ & $-1.96$ & $-2.62$ & $-0.67$ & $-1.32$ & $0.55$ & $0.66$ & $0.62$ & $0.76$ & $0.48$ & $\mathbf{0.56}$ & $2.06$ & $2.70$ & $1.01$ & $1.40$ & $0.86$ & $0.04$ & $0.00$ & $0.78$ & $0.05$ & $\mathbf{0.876}$ & $0.853$ & $0.850$ & $0.869$ & $0.869$ & $0.875$ \\
$500$ & logistic & $N(0,9)$ & $\mathbf{+0.05}$ & $-1.99$ & $-2.60$ & $-0.58$ & $-1.30$ & $0.36$ & $0.38$ & $0.35$ & $0.51$ & $0.28$ & $\mathbf{0.36}$ & $2.03$ & $2.63$ & $0.77$ & $1.33$ & $0.94$ & $0.00$ & $0.00$ & $0.84$ & $0.00$ & $\mathbf{0.878}$ & $0.856$ & $0.852$ & $0.874$ & $0.872$ & $0.877$ \\
$1000$ & logistic & $N(0,9)$ & $\mathbf{+0.01}$ & $-2.02$ & $-2.59$ & $-0.55$ & $-1.30$ & $0.25$ & $0.27$ & $0.25$ & $0.38$ & $0.20$ & $\mathbf{0.25}$ & $2.03$ & $2.60$ & $0.67$ & $1.31$ & $0.96$ & $0.00$ & $0.00$ & $0.75$ & $0.00$ & $\mathbf{0.878}$ & $0.856$ & $0.852$ & $0.876$ & $0.873$ & $\mathbf{0.878}$ \\
$200$ & logistic & min-EV & $\mathbf{+0.10}$ & $-1.95$ & $-2.60$ & $-0.63$ & $-1.31$ & $0.60$ & $0.55$ & $0.48$ & $0.63$ & $0.36$ & $\mathbf{0.61}$ & $2.03$ & $2.64$ & $0.89$ & $1.36$ & $0.84$ & $0.01$ & $0.00$ & $0.82$ & $0.01$ & $0.875$ & $0.854$ & $0.851$ & $0.872$ & $0.871$ & $\mathbf{0.876}$ \\
$500$ & logistic & min-EV & $\mathbf{+0.03}$ & $-2.01$ & $-2.60$ & $-0.56$ & $-1.30$ & $0.37$ & $0.32$ & $0.28$ & $0.41$ & $0.22$ & $\mathbf{0.37}$ & $2.04$ & $2.61$ & $0.69$ & $1.32$ & $0.92$ & $0.00$ & $0.00$ & $0.81$ & $0.00$ & $0.877$ & $0.855$ & $0.852$ & $0.875$ & $0.873$ & $\mathbf{0.878}$ \\
$1000$ & logistic & min-EV & $\mathbf{+0.00}$ & $-2.04$ & $-2.60$ & $-0.53$ & $-1.30$ & $0.24$ & $0.23$ & $0.21$ & $0.29$ & $0.15$ & $\mathbf{0.24}$ & $2.05$ & $2.61$ & $0.61$ & $1.31$ & $0.98$ & $0.00$ & $0.00$ & $0.64$ & $0.00$ & $\mathbf{0.878}$ & $0.856$ & $0.852$ & $0.876$ & $0.873$ & $\mathbf{0.878}$ \\
$200$ & exp & $N(0,9)$ & $\mathbf{+0.06}$ & $-1.95$ & $-2.61$ & $-0.69$ & $-1.33$ & $0.54$ & $0.63$ & $0.63$ & $0.76$ & $0.46$ & $\mathbf{0.55}$ & $2.05$ & $2.69$ & $1.03$ & $1.41$ & $0.88$ & $0.02$ & $0.00$ & $0.79$ & $0.04$ & $\mathbf{0.872}$ & $0.848$ & $0.845$ & $0.865$ & $0.865$ & $0.866$ \\
$1000$ & exp & $N(0,9)$ & $\mathbf{+0.01}$ & $-2.02$ & $-2.57$ & $-0.58$ & $-1.32$ & $0.24$ & $0.26$ & $0.24$ & $0.39$ & $0.20$ & $\mathbf{0.24}$ & $2.03$ & $2.58$ & $0.70$ & $1.33$ & $0.96$ & $0.00$ & $0.00$ & $0.71$ & $0.00$ & $\mathbf{0.874}$ & $0.850$ & $0.847$ & $0.872$ & $0.868$ & $0.869$ \\
$1000$ & exp & min-EV & $\mathbf{+0.00}$ & $-2.00$ & $-2.58$ & $-0.55$ & $-1.33$ & $0.25$ & $0.23$ & $0.20$ & $0.30$ & $0.16$ & $\mathbf{0.25}$ & $2.01$ & $2.59$ & $0.62$ & $1.34$ & $0.97$ & $0.00$ & $0.00$ & $0.63$ & $0.00$ & $\mathbf{0.876}$ & $0.852$ & $0.848$ & $0.874$ & $0.871$ & $0.872$ \\
\hline
\end{tabular}}

\vspace{0.4em}
\begin{minipage}{\textwidth}\footnotesize\raggedright
Stat groups: Bias, empirical standard error (ESE), root mean squared error (RMSE), and empirical $95\%$ coverage probability (CP) of $\hat{X}_c$; the test-set Harrell's $C$-index is reported as a global summary of model discrimination. Methods (sub-columns): MCE; Sp (spline); Q (quadratic); Sp$\!\times\!\!Z$ (spline with $Z$ interaction); Q$\!\times\!\!Z$ (quadratic with $Z$ interaction); Or (oracle Cox using true $X_c(\bZ)$, reported only for $C$-index since the oracle does not estimate $X_c$). Results based on $1{,}000$ replications per scenario at $\approx 30\%$ censoring; true $X_c = 1.5$ in every scenario. Boldface marks the lowest $|\text{Bias}|$ and lowest RMSE (across the five feasible methods) and the highest $C$-index (across all six methods) within each scenario.
\end{minipage}
\end{sidewaystable}

\clearpage
\begin{figure}[!htbp]
    \centering
    \includegraphics[width=\textwidth]{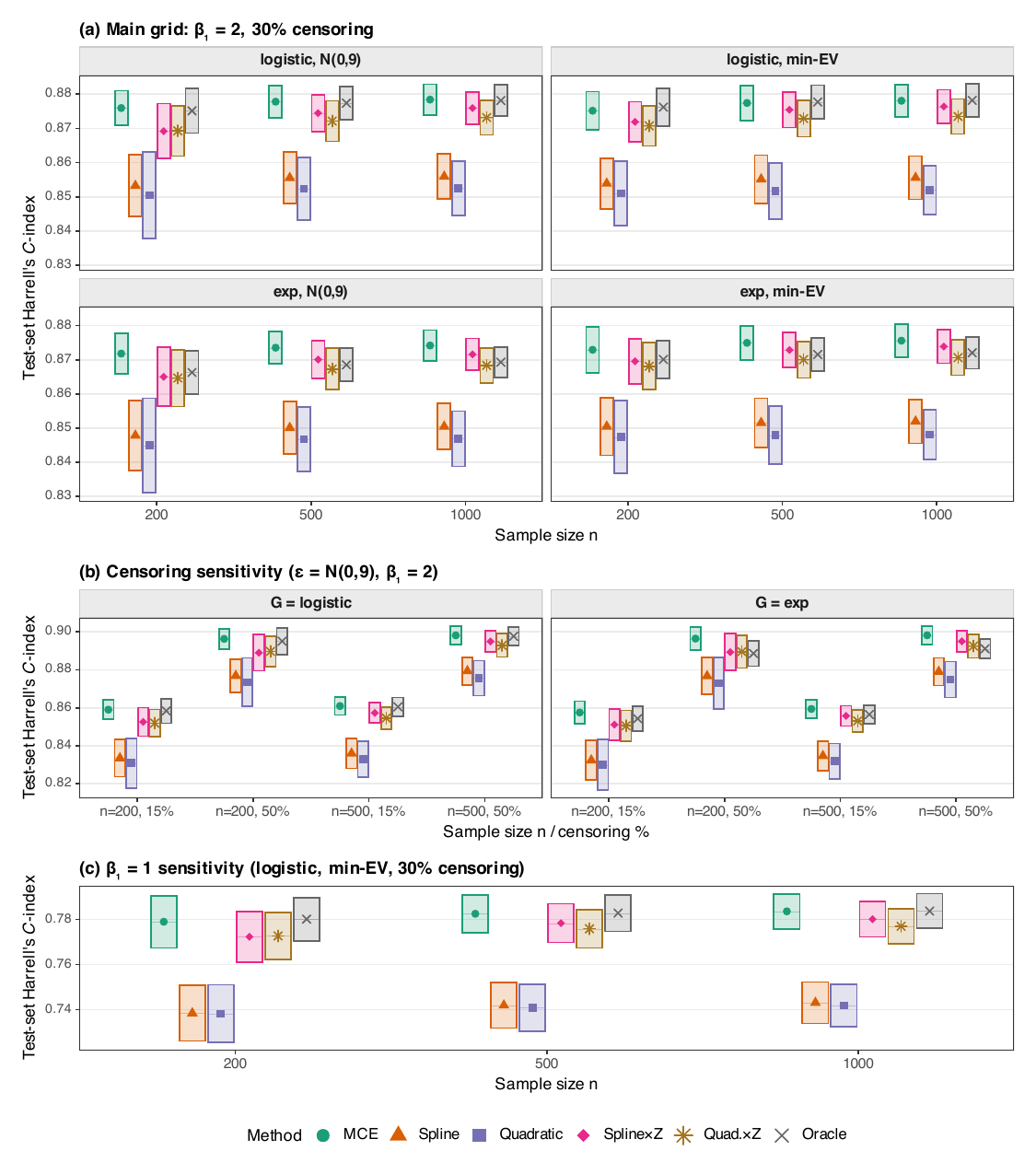}
    \caption{\small\linespread{1.0}\selectfont Held-out test-set Harrell's $C$-index across all $23$ simulation scenarios and six methods (MCE, four Cox alternatives, and the oracle Cox model using the true subject-specific $X_c(\bZ)$). Each marker is the across-replicate mean and the surrounding box spans $\pm 1.96\times$ between-replicate standard deviation. Panel arrangement (a)--(c) matches Figure~\ref{fig:sim-xc}. The compressed $C$-index gap between methods, despite the large differences in $\hat X_c$ accuracy shown in Figure~\ref{fig:sim-xc}, illustrates that the $C$-index alone is relatively insensitive to differences in risk-model structure; nonetheless, the MCE achieves the highest $C$-index across the great majority of scenarios and matches or exceeds the oracle under normal error. Results are based on $1{,}000$ replications per scenario, with $C$-index evaluated on an independent test set of $5{,}000$ observations.}
    \label{fig:sim-cindex}
\end{figure}

\begin{table}[!htbp]
\centering
\caption{Held-out test-set Harrell's $C$-index across all $23$ simulation scenarios, by estimation method. Each value is the mean across $1{,}000$ replications on an independent test set of $5{,}000$ observations. Boldface marks the highest $C$-index in each row.}\label{tab:supp-cindex-all}
\vspace{-6pt}
\resizebox{\textwidth}{!}{%
\begin{tabular}{ccccc ccccccc}
\hline
$n$ & $G$ & $\epsilon$ & $\beta_1$ & cens. & MCE & Spline & Quad. & Spline$\times Z$ & Quad.$\times Z$ & Oracle \\
\hline
$200$ & logistic & $N(0,9)$ & $2$ & $15\%$ & $\mathbf{0.859}$ & $0.834$ & $0.831$ & $0.853$ & $0.852$ & $0.858$ \\
$200$ & logistic & $N(0,9)$ & $2$ & $30\%$ & $\mathbf{0.876}$ & $0.853$ & $0.850$ & $0.869$ & $0.869$ & $0.875$ \\
$200$ & logistic & $N(0,9)$ & $2$ & $50\%$ & $\mathbf{0.896}$ & $0.877$ & $0.873$ & $0.889$ & $0.890$ & $0.895$ \\
$500$ & logistic & $N(0,9)$ & $2$ & $15\%$ & $\mathbf{0.861}$ & $0.836$ & $0.833$ & $0.857$ & $0.855$ & $\mathbf{0.861}$ \\
$500$ & logistic & $N(0,9)$ & $2$ & $30\%$ & $\mathbf{0.878}$ & $0.856$ & $0.852$ & $0.874$ & $0.872$ & $0.877$ \\
$500$ & logistic & $N(0,9)$ & $2$ & $50\%$ & $\mathbf{0.898}$ & $0.879$ & $0.876$ & $0.895$ & $0.893$ & $\mathbf{0.898}$ \\
$1000$ & logistic & $N(0,9)$ & $2$ & $30\%$ & $\mathbf{0.878}$ & $0.856$ & $0.852$ & $0.876$ & $0.873$ & $\mathbf{0.878}$ \\
$200$ & logistic & min-EV & $2$ & $30\%$ & $0.875$ & $0.854$ & $0.851$ & $0.872$ & $0.871$ & $\mathbf{0.876}$ \\
$500$ & logistic & min-EV & $2$ & $30\%$ & $0.877$ & $0.855$ & $0.852$ & $0.875$ & $0.873$ & $\mathbf{0.878}$ \\
$1000$ & logistic & min-EV & $2$ & $30\%$ & $\mathbf{0.878}$ & $0.856$ & $0.852$ & $0.876$ & $0.873$ & $\mathbf{0.878}$ \\
$200$ & exp & $N(0,9)$ & $2$ & $15\%$ & $\mathbf{0.858}$ & $0.832$ & $0.830$ & $0.851$ & $0.851$ & $0.854$ \\
$200$ & exp & $N(0,9)$ & $2$ & $30\%$ & $\mathbf{0.872}$ & $0.848$ & $0.845$ & $0.865$ & $0.865$ & $0.866$ \\
$200$ & exp & $N(0,9)$ & $2$ & $50\%$ & $\mathbf{0.896}$ & $0.877$ & $0.873$ & $0.889$ & $0.890$ & $0.889$ \\
$500$ & exp & $N(0,9)$ & $2$ & $15\%$ & $\mathbf{0.859}$ & $0.835$ & $0.832$ & $0.856$ & $0.853$ & $0.856$ \\
$500$ & exp & $N(0,9)$ & $2$ & $30\%$ & $\mathbf{0.874}$ & $0.850$ & $0.847$ & $0.870$ & $0.867$ & $0.869$ \\
$500$ & exp & $N(0,9)$ & $2$ & $50\%$ & $\mathbf{0.898}$ & $0.879$ & $0.875$ & $0.895$ & $0.892$ & $0.891$ \\
$1000$ & exp & $N(0,9)$ & $2$ & $30\%$ & $\mathbf{0.874}$ & $0.850$ & $0.847$ & $0.872$ & $0.868$ & $0.869$ \\
$200$ & exp & min-EV & $2$ & $30\%$ & $\mathbf{0.873}$ & $0.850$ & $0.847$ & $0.870$ & $0.868$ & $0.870$ \\
$500$ & exp & min-EV & $2$ & $30\%$ & $\mathbf{0.875}$ & $0.852$ & $0.848$ & $0.873$ & $0.870$ & $0.872$ \\
$1000$ & exp & min-EV & $2$ & $30\%$ & $\mathbf{0.876}$ & $0.852$ & $0.848$ & $0.874$ & $0.871$ & $0.872$ \\
$200$ & logistic & min-EV & $1$ & $30\%$ & $0.779$ & $0.738$ & $0.738$ & $0.772$ & $0.773$ & $\mathbf{0.780}$ \\
$500$ & logistic & min-EV & $1$ & $30\%$ & $0.782$ & $0.742$ & $0.741$ & $0.778$ & $0.776$ & $\mathbf{0.783}$ \\
$1000$ & logistic & min-EV & $1$ & $30\%$ & $\mathbf{0.784}$ & $0.743$ & $0.742$ & $0.780$ & $0.777$ & $\mathbf{0.784}$ \\
\hline
\end{tabular}}
\end{table}

\end{document}